\documentclass[twoside,12pt,onecolumn,draftcls]{IEEEtran}
\usepackage{amsmath,amsfonts, amssymb}
\usepackage[dvips]{graphicx}

\usepackage{algorithm}
\usepackage{algorithmic}
\usepackage{cite}

\usepackage{times}
\usepackage{epsfig}
\usepackage{amstext}
\usepackage{latexsym}
\usepackage{color}
\usepackage{ifthen}
\usepackage{multirow}
\usepackage{subfigure}
\usepackage{pstricks}

\newtheorem{theorem}{Theorem}    
\newtheorem{lemma}{Lemma}      
\newcommand {\beq}{\begin{equation}}
\newcommand {\eeq}{\end{equation}}
\newcommand {\beqn}{\begin{eqnarray}}
\newcommand {\eeqn}{\end{eqnarray}}

\def\inr{\mathsf{INR}}
\def\snr{\mathsf{SNR}}
\def\C{\mathcal{C}}
\def\R{\mathcal{R}}
\def\Pr{\mathcal{P}}
\def\CNE{\mathcal{C}_{\rm NE}}
\def\vx{\mathbf{x}}
\def\vy{\mathbf{y}}
\def\vs{\mathbf{s}}
\def\vu{\mathbf{u}}
\def\vv{\mathbf{v}}
\def\vR{\mathbf{R}}

\def\tvy{\tilde{\mathbf{y}}}

\def\eps{\epsilon}

\def\B{\mathcal{B}}

\def\vz{\mathbf{z}}
\def\CN{\mathcal{CN}}
\def\tvu{\tilde{\mathbf{u}}}

\def\tvz{\tilde{\mathbf{z}}}

\begin{document}
\title{Shannon Meets Nash \\ on the Interference Channel}
\author{Randall A. Berry and David N. C. Tse
\thanks{R.A. Berry is with the Dept.~of Electrical
Engineering and Computer Science, Northwestern University,
Evanston, IL  60208, USA, email: {\tt rberry@eecs.northwestern.edu}.
D.N.C. Tse is with Wireless Foundations, University of California, 
Berkeley, Berkeley, CA 94720, USA, e-mail: {\tt dtse@eecs.berkeley.edu}}
\thanks{Preliminary versions of this work were presented in part at the 
2008 International Symposium on Information Theory and 
the 2009 Information Theory Workshop.}}

\date{July 10, 2010}
\maketitle

\begin{abstract}

The interference channel is the simplest communication scenario
where multiple autonomous users compete for shared resources. We
combine game theory and information theory to define a notion of
a Nash equilibrium region of the interference channel. The notion is
game theoretic: it captures the selfish behavior of each user as
they compete. The notion is also information theoretic: it allows
each user to use arbitrary communication strategies as it optimizes
its own performance. We give an exact characterization of the Nash
equilibrium region of the two-user linear deterministic interference
channel and an approximate characterization of the Nash equilibrium
region of the two-user Gaussian interference channel to within 1 bit/s/Hz..

\end{abstract}

\section{Introduction}

Information theory deals with the fundamental limits of
communication. In network information theory, an object of central
interest is the {\em capacity region} of the network: it is the set
of all rate tuples of the users in the network simultaneously
achievable by  optimizing their communication strategies. Implicit
in the definition is that users optimize their communication
strategies {\em cooperatively}. This may not be a realistic
assumption if users are selfish and are only interested in
maximizing their own benefit. Game theory provides a notion of {\em
Nash equilibrium} to characterize system operating points, which are
stable under such selfish behavior. In this paper, we define and
explore an information theoretic Nash equilibrium region as the game
theoretic counterpart of the capacity region of a network. While the
Nash equilibrium region is naturally a subset of the (cooperative)
capacity region, in general not all points in the capacity region
are Nash equilibria. The research question is then to characterize
the Nash equilibrium region given a network and a model of the
channels.

The two-user interference channel (IC) is perhaps the simplest
communication scenario to study this problem. Here two
point-to-point communication links interfere with each other through
cross-talks. Each transmitter has an independent message intended
only for the corresponding receiver. The capacity region of this
channel is the set of all simultaneously achievable rate pairs
$(R_1,R_2)$ in the two interfering links, and characterizes the
fundamental tradeoff between the performance achievable on the two
links in face of interference.

In the cooperative setting, the users jointly choose encoding and
decoding schemes to achieve a rate pair $(R_1,R_2)$. In the game
theoretic setting, on the other hand, we study the case where each
user individually chooses an encoding/decoding scheme in order to
maximize his own transmission rate. The two users can be viewed as
playing a {\em non-cooperative} game, where a user's strategy is its
encoding/decoding scheme and its payoff is its reliable rate. A Nash
equilibrium (NE) is a pair of strategies for which there is no incentive
for either user to unilaterally change its strategy to  improve its
own rate. These are incentive-compatible operating points. The Nash
equilibrium region of the IC is the set of all reliable rate pairs
each of which can be achieved at some NE. Our focus is on a
``one-shot'' game formulation in which each player has full
information, i.e. both players know the channel statistics, the
actions chosen by each player, as well as their pay-off function.

A particular IC we focus on in this paper is the two-user Gaussian
IC shown in Figure~\ref{fig:model}. 
This is a basic model in wireless and wireline
 channels (such as DSL). Game theoretic approaches for the Gaussian IC
have been studied before, e.g.~\cite{YC00,Chu03,LZ06,EPT07}. However,
there are two key assumptions in these works: 1) the class of
encoding strategies are constrained to use random Gaussian
codebooks; 2) the decoders are restricted to treat the interference
as Gaussian noise and are hence sub-optimal. Because of these
restrictions, the formulation in these works are not
information-theoretic in nature. For example, a Nash equilibrium
found under these assumptions may no longer be an equilibrium if
users can adopt a different encoding or decoding strategy.

\begin{figure}[htb]
\centerline{ \psfig{figure=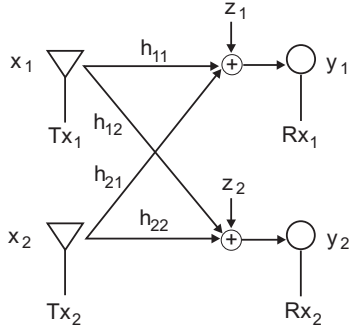,width=1.8in}} \caption{
Two-user Gaussian interference channel. } \label{fig:model}
\vspace{-.1in}
\end{figure}

In this paper, we make three contributions. First, we give a precise
formulation of an information theoretic NE on general ICs, where the
users are allowed to use {\em any} encoding and decoding strategies.
Second, we analyze the NE region of the two-user {\em linear
deterministic} IC \cite{BT08}. This type of deterministic channel
model was first proposed by \cite{ADT07} in the analysis of
Gaussian relay networks, and the deterministic IC  has been shown to
be a good approximation of the Gaussian IC in \cite{BT08}. For this
IC, we give a simple exact characterization of the NE region. At
each of the rate pairs in the NE region, we provide explicit coding
schemes that achieve the rate pair and such that no user
has any incentive to deviate to improve its own rate. Somewhat
surprisingly, we find that in all cases, there are always Nash
equilibria which are {\em efficient}, i.e., they lie on the maximum
sum-rate boundary of the capacity region. In particular, for
channels with symmetrical channel gains, the symmetric rate point on
the capacity region boundary is always a Nash equilibrium. Our third
contribution is to use these insights to approximate the NE region
of the Gaussian IC to within $1$ bit/s/Hz. This result parallels the
recent characterization of 
the (cooperative) capacity region of the same channel to within $1$
bit \cite{ETW07}.

\section{Problem Formulation}
\label{sec:formulation}

Let us now formally define the communication situation for general
interference channels. In subsequent sections, we will specialize to
specific classes of interference channels.

Communication starts at time $0$. User $i$ communicates by coding
over blocks of length $N_i$ symbols, $i=1,2$. Transmitter $i$ sends
on block $k$ information bits $b^{(k)}_{i1}, \ldots,
b_{i,L_i}^{(k)}$ by transmitting a codeword denoted by $\vx_i^{(k)}
=[\vx^{(k)}_i(1),\ldots, \vx^{(k)}_i(N_i)]$. All the information
bits are equally probable and independent of each other. Receiver
$i$ observes on each block an output sequence through the
interference channel, which specifies a stochastic mapping from the
input sequences of user 1 and 2 to the output sequences of user 1
and $2$. Given the observed sequence $\{\vy^{(k)}_i =
[\vy^{(k)}_i(1), \ldots, \vy^{(k)}_i(N_i)], k =1,2, \ldots, \}$,
receiver $i$ generate guesses $\hat{b}^{(k)}_{i\ell}$ for each of
the information bit. Without loss of generality, we will assume that
each receiver $i$ performs maximum-likelihood decoding on each bit,
i.e. chooses $\hat{b}^{(k)}_i$ that maximizes the {\em a posterior
probability} of the observed sequence $\vy^{(1)}_i, \vy^{(2)}_i,
\ldots$ given the transmitted bit $b^{(k)}_{i\ell}$.

Note that the communication scenario we defined here is more general
than the one usually used in multiuser information theory, as we
allow the two users to code over different block lengths. However,
such generality is necessary here,  since even though the two users
may agree {\em a priori} on a common block length, a selfish user
may unilaterally decide to choose a different block length during
the actual communication process.

A {\em strategy} $s_i$ of user $i$ is defined by its message
encoding, which  we assume to be the same on every block and
involves:
\begin{itemize}
\item the number of information bits $B_i$ and the block length $N_i$ of the codewords,
\item the codebook $\C_i$, the set of codewords employed by
transmitter $i$,
\item the encoder $f_i : \{1, \ldots, 2^{B_i}\} \times \Omega_i
\rightarrow \C_i$, that maps on each block $k$ the message
$m_i^{(k)} := (b_{i1}^{(k)}, \ldots b_{i,B_i}^{(k)})$ to a transmitted codeword $\vx^{(k)}_i =
f_i(m^{(k)}_i,\omega^{(k)}_i) \in \C_i$,
\item the rate of the code, $R_i(s_i) = B_i / N_i$.
\end{itemize}

A strategy $s_1$ of user $1$ and $s_2$ of user $2$ jointly
determines the average bit error probabilities $p^{(k)}_i := \frac{1}{B_i}
\sum_{\ell=1}^{B_i} \Pr ( \hat{b}^{(k)}_{i\ell} \not =
b^{(k)}_{i\ell})$, $i=1,2.$\footnote{Average bit error 
probabilities are more meaningful than codeword error probabilities
in a setting, such as ours, where users can vary 
the blocklength they are using.} Note that if the two users use different
block lengths, the error probability could vary from block to block
even though each user uses the same encoding for all the blocks.
However, if they use the same block length, then the error
probability is the same across the blocks for a user, which we will
denote by $p_i$ for user $i$. 

The encoder of each transmitter $i$
may employ a stochastic mapping from the message to the transmitted
codeword; $\omega^{(k)}_i \in \Omega_i$ represents the randomness in
that mapping. We assume that this randomness is independent between
the two transmitters and across different blocks. Furthermore, we assume 
that each transmitter and its corresponding receiver have access to a 
source of {\it common randomness}, so that the realization $\omega^{(k)}_i$ is known at both transmitter $i$ and receiver $i$, but not at 
the other receiver or transmitter.\footnote{Such common randomness is not
needed for many of the results in the paper, but allowing for it
simplifies our presentation.}

For a given error probability threshold $\eps >0$, we define an
$\eps$-interference channel game as follows. Each user $i$
chooses a strategy $s_i$, $i=1,2,$ and receives a pay-off of
\[
\pi_i(s_1,s_2) =\begin{cases}
R(s_i), & \text{ if $p_{i}^{(k)}(s_1,s_2) \leq \epsilon$, $\forall k$},\\
0, & \text{ otherwise.}
\end{cases}
\]
In other words, a user's pay-off is equal to the rate of the code
provided that the probability of error is no greater than $\epsilon$.
A strategy pair $(s_1,s_2)$ is defined to be {\it $(1-\eps)$-reliable}
provided that they result in an error probability $p_i(s_1,s_2)$ of
less than $\eps$ for $i=1,2$. An $(1-\eps)$-reliable pair of strategies
is said to achieve the rate-pair $(R(s_1),R(s_2))$.

For an $\eps$-game, a strategy pair $(s_1^*,s_2^*)$
is a Nash equilibrium (NE) if neither user
can unilaterally deviate and improve their pay-off, i.e.~if for each
user $i=1,2$, there is no other strategy $s_i$ such that\footnote{In this
paper, we use the convention that $j$ always denotes the other user from $i$.}
 $\pi_i(s_i,s_j^*) > \pi_i(s_i^*,s_j^*).$
If user $i$ attempts to transmit at a higher rate than what he is receiving in a
Nash equilibrium and user $j$ does not change her strategy, then user
$i$'s error probability must be greater than $\eps$.

Similarly, a strategy pair $(s_1^*,s_2^*)$ is an
 $\eta$-Nash equilibrium\footnote{In the game theoretic literature, this is often referred to as an
$\eps$-Nash equilibrium or simply an $\eps$-equilibrium for a game
\cite[page 143]{Myerson91}.} ($\eta$-NE) of an $\eps$-game if neither user can
unilaterally deviate and improve their pay-off by more than $\eta$,
i.e.~if for each user $i$, there is no other strategy $s_i$ such that
$\pi_i(s_i,s_j^*) > \pi_i(s_i^*,s_j^*) + \eta.$ Note that when a user
deviates, it does not care about the reliability of the other user but
only its own reliability. So in the above definitions $(s_i,s_j^*)$ is
not necessarily $(1-\eps)$-reliable.

Given any $\bar{\eps} >0$, the capacity region $\C$ of the interference
channel is the closure of the set of all rate pairs $(R_1,R_2)$ such
that
for every $\eps \in (0,\bar{\eps})$, there exists a $(1-\eps)$-reliable
strategy pair $(s_1,s_2)$ which achieves the rate pair $(R_1,R_2)$.
The {\it Nash equilibrium region} $\CNE$ of the interference channel
is the closure of the set of rate pairs $(R_1,R_2)$ such that for every
$\eta>0$, there exists a $\bar{\eps}>0$ (dependent on $\eta$) so that if
$\eps \in (0,\bar{\eps})$, there exists a $(1-\eps)$-reliable strategy pair
$(s_1,s_2)$ that achieves the rate-pair $(R_1,R_2)$ and is a $\eta$-NE.
Clearly, $\CNE \subset \C$. 

First, we make a few comments about the
definition of $\CNE$.  In this definition, the
parameter $\bar{\eps}$ is introduced so that
$(1-\eps)$-reliable
strategy pairs need only exist for ``small enough'' values of $\eps$.
 In the definition of the capacity region for
the interference channel this constraint is not needed, i.e. the
region is equally well defined by requiring the given
conditions to hold for any $\epsilon >0$ (since, clearly if a pair of
strategies are  $(1-\epsilon)$-reliable, they are
also $(1-\tilde{\epsilon})$-reliable for all $\tilde{\epsilon}
>\epsilon$). However, when defining $\CNE$, this condition is
important. In particular a pair of strategies can be an
$\eta$-NE for an $\eps$-game, but not an $\eta$-NE for an
$\tilde{\eps}$-game for  $\tilde{\epsilon} >\epsilon$, since
increasing the error probability
threshold enlarges the set of possible deviations an agent may make.
As an extreme example, consider the case where $\epsilon = 1$, in
which case each agent can achieve an arbitrarily high pay-off
regardless of the action of the other user and so no
$\eta$-NE exists. Thus, if we required our definition to hold
for any $\epsilon>0$, $\CNE$ would be empty.

Next, we turn to the use of $\eta$-NE in the definition.
A more natural approach would be to instead simply use NE.
In other words, define $\CNE$ to be
the closure of the rate pairs $(R_1,R_2)$ such that for any
$\eps$ small enough, that there exists a $(1-\eps)$-reliable strategy pair
$(s_1,s_2)$ which achieves the rate-pair $(R_1,R_2)$ and is a NE of
a $\eps$-game. The difficulty with this is that to determine such a
NE essentially requires one to find a particular scheme that achieves
the optimal rate for a given non-zero error probability. Finding such
a scheme that is extremely difficult and in general an open
problem.\footnote{Moreover, it is not even clear if there exists such
  a scheme, i.e.~a scheme that
 achieves the supremum of the rates over all $1-\eps$ reliable schemes.}
By introducing the slack $\eta$, these difficulties are
removed. Moreover,
since we require that this definition hold for all $\eta>0$, this
slack can be made arbitrarily small.

Finally, we would like to comment on the use of different block
lengths in our definitions. First, it can argued that if there is a
$(1-\eps)$-reliable strategy pair $(s_1,s_2)$ that achieves a rate
pair $(R_1,R_2)$ using codes of block lengths $N_1,N_2$, then there
exists a $(1-\eps)$ strategy pair that achieves the same rate pair
but with each user using the same block length. This follows by
considering using ``super-blocks'' of length $N$, where $N$ is the
least common multiple of $N_1$ and $N_2$. Over these super-blocks
the users can be viewed as using two equal-length codes. The error
probabilities, being the average bit error probabilities now across
longer blocks, remain less than $\eps$. This means that in computing
the capacity region $\C$, we can without loss of generality consider
only strategies in which both users use the same block lengths.
Also, in the Nash equilibrium definitions, we can without loss of
generality assume that in the nominal strategy, the two users use
the same block length (although each user is allowed to deviate
using another strategy of a different block length.).

\section{The Linear-deterministic IC}

\subsection{Deterministic Channel Model}
\label{sec:DeterministicChannel}

Let us now focus on a specific interference channel model: a
linear deterministic channel model analogous to the Gaussian channel. This
channel was first introduced in \cite{ADT07}. We begin by describing
the deterministic channel model for the point-to-point AWGN channel,
and then the two-user multiple-access channel. After understanding
these examples, we present the deterministic interference channel.

Consider first the model for the point-to-point channel (see
Figure~\ref{fig:p2pDeterministic}). The real-valued channel input is
written in base 2; the signal---a vector of bits---is interpreted as
occupying a succession of levels:
$$x=0.b_1 b_2 b_3 b_4 b_5\dots\,.$$
The most significant bit coincides with the highest level, the least
significant bit with the lowest level. The levels attempt to capture
the notion of \emph{signal scale}; a level corresponds to a unit of
power in the Gaussian channel, measured on the dB scale. Noise is
modeled in the deterministic channel by truncation. Bits of smaller
order than the noise are lost. Note that the number of bits above
the noise floor correspond to $\log_2 \snr$, where $\snr$ is the
signal-to-noise ratio of the corresponding Gaussian channel.

\begin{figure}
\begin{centering}
\psset{unit=0.9mm,linewidth=.3pt,arrowlength=1,arrowinset=0}
\begin{center}
\begin{pspicture}(-10,3)(40,24)

\pscircle(0,4){1.9}
\pscircle(0,8){1.9}
\pscircle(0,12){1.9}
\pscircle(0,16){1.9}
\pscircle(0,20){1.9}

\rput(0,4){\footnotesize $b_5$}
\rput(0,8){\footnotesize $b_4$}
\rput(0,12){\footnotesize $b_3$}
\rput(0,16){\footnotesize $b_2$}
\rput(0,20){\footnotesize $b_1$}

\psline{->}(1.5,16)(28.5,12)
\psline{->}(1.5,8)(28.5,4)
\psline{->}(1.5,12)(28.5,8)
\psline{->}(1.5,20)(28.5,16)

\pspolygon(-2,1)(-6,1)(-6,23)(-2,23)

\rput(30,0){
\pscircle(0,4){1.9}
\pscircle(0,8){1.9}
\pscircle(0,12){1.9}
\pscircle(0,16){1.9}
\pscircle(0,20){1.9}
\pspolygon(2,1)(6,1)(6,23)(2,23)
\rput(0,-4){\rput(0,8){\footnotesize $b_4$}
\rput(0,12){\footnotesize $b_3$}
\rput(0,16){\footnotesize $b_2$}
\rput(0,20){\footnotesize $b_1$}}
}

\psline[linestyle=dashed](0,1)(30,1)
\psline(1.5,4)(20.25,1)

\uput[l](-6,12){Tx}
\uput[r](36,12){Rx}
\uput[d](15,0){noise}
\end{pspicture}
\end{center}
\caption{The deterministic model for the point-to-point Gaussian
channel. Each bit of the input occupies a signal level. Bits of
lower significance are lost due to noise.}
\label{fig:p2pDeterministic}
\end{centering}
\end{figure}
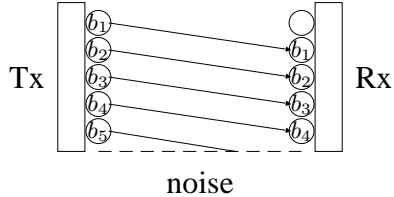

The deterministic multiple-access channel is constructed similarly
to the point-to-point channel (Figure~\ref{fig:deterministicMAC}).
To model the super-position of signals at the receiver, the bits
received on each level are added {\em modulo two}. Addition modulo
two, rather than normal integer addition, is chosen to make the
model more tractable. As a result, the levels do not interact with
one another.

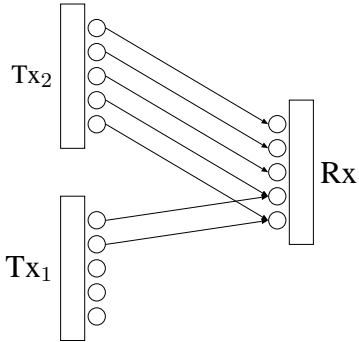
\begin{figure}
\begin{centering}
\psset{unit=.8mm,linewidth=.3pt,arrowlength=1.2,arrowinset=0,labelsep=5pt}
\begin{center}
\begin{pspicture}(-1,18)(37,68)

\rput(0,14){\pscircle(0,0){1.5}
\pscircle(0,4){1.5}
\pscircle(0,8){1.5}
\pscircle(0,12){1.5}
\pscircle(0,16){1.5}
\pspolygon(-2,-4)(-6,-4)(-6,20)(-2,20)

\uput[l](-5,8){$ {\text{Tx}_1}$}
}

\psline{->}(1.5,30)(28.5,34)
\psline{->}(1.5,26)(28.5,30)

\uput[r](35,38){$ \text{Rx}$}


\rput(0,46){
\pscircle(0,0){1.5}
\pscircle(0,4){1.5}
\pscircle(0,8){1.5}
\pscircle(0,12){1.5}
\pscircle(0,16){1.5}
\pspolygon(-2,-4)(-6,-4)(-6,20)(-2,20)
}
\psline{->}(1.5,46)(28.5,30)
\psline{->}(1.5,50)(28.5,34)
\psline{->}(1.5,54)(28.5,38)
\psline{->}(1.5,58)(28.5,42)
\psline{->}(1.5,62)(28.5,46)

\uput[l](-5,54){$\small \text{Tx}_2$}

\rput(30,30){
\pscircle(0,0){1.5}
\pscircle(0,4){1.5}
\pscircle(0,8){1.5}
\pscircle(0,12){1.5}
\pscircle(0,16){1.5}
\pspolygon(2,-4)(6,-4)(6,20)(2,20)
}

\end{pspicture}
\end{center} 
\caption{The deterministic model for the Gaussian multiple-access
channel. Incoming bits on the same level are added modulo two at the
receiver.} \label{fig:deterministicMAC}
\end{centering}
\end{figure}

We proceed with the deterministic interference channel model
(Fig.~\ref{fig:ICrectangles2}). There are two transmitter-receiver
pairs (links), and as in the Gaussian case, each transmitter wants
to communicate only with its corresponding receiver.  The signal
from transmitter $i $, as observed at receiver $j$, is scaled by a
nonnegative integer gain $a_{ji}=2^{n_{ji}} $ (equivalently, the
input column vector is shifted up by $n_{ji}$). At each time $t $,
the input and output, respectively, at link $i$ are
$\vx_i(t),\vy_i(t)\in \{0,1\}^q$, where $q=\max_{ij} n_{ij}$. Note
that $n_{ii}$ corresponds to $\log_2 \snr_i$ and $n_{ji}$
corresponds to $\log_2 \inr_{ji}$, where $\snr_i$ is the
signal-to-noise ratio of link $i$ and $\inr_{ji}$ is the
interference-to-noise ratio at receiver $j$ from transmitter $i$ in
the corresponding Gaussian interference channel.

The channel output at receiver $i$ is given by
\begin{equation}
\label{eq:channel} \vy_i(t)=
\mathbf{S}^{q-n_{i1}}\vx_1(t)+\mathbf{S}^{q-n_{i2}}\vx_2(t),\end{equation}
where summation and multiplication are in the binary field and
$\mathbf{S}$ is a $q\times q$ shift matrix,
\begin{equation}
\mathbf{S}=\left(
\begin{matrix} 0 & 0 & 0 & \cdots & 0 \cr 1 & 0 & 0 & \cdots & 0 \cr 0 & 1 & 0 & \cdots & 0 \cr \vdots & & & \ddots & \vdots \cr 0 & \cdots & 0 & 1 & 0 \end{matrix}
\right).
\end{equation}
If the inputs $\vx_i(t)$ are written as a binary number $x_i$, the
channel can equivalently be written as
\begin{align*}
  y_1&=\lfloor a_{11}x_1+ a_{12} x_2\rfloor \\
  y_2&=\lfloor a_{21}x_1+ a_{22}x_2 \rfloor \, ,
\end{align*}
where addition is performed on each bit (modulo two) and $\lfloor\,
\cdot \, \rfloor$ is the integer-part function.

In our analysis, it will be helpful to consult a different style of
figure, as shown on the right-hand side of
Fig.~\ref{fig:ICrectangles2}. This shows only the perspective of
each receiver. Each incoming signal is shown as a column vector,
with the highest element corresponding to the most significant bit
and the portion below the noise floor truncated. The observed signal
at each receiver is the modulo 2 sum of the elements on each level.
\begin{figure}
\begin{centering}
\psset{unit=.71mm,linewidth=.3pt,arrowlength=1.2,
arrowinset=0,labelsep=3pt}
\begin{center}
\begin{pspicture}(16,37)(80,73)

\rput(0,26){
\pscircle(0,4){1.5}
\pscircle(0,8){1.5}
\pscircle(0,12){1.5}
\pscircle(0,16){1.5}
\psline{->}(1.5,16)(24.5,16)
\psline{->}(1.5,8)(24.5,8)
\psline{->}(1.5,12)(24.5,12)
\psline{->}(1.5,4)(24.5,4)
\pspolygon(-2,0)(-6,0)(-6,20)(-2,20)
}

\uput[l](-5,36){$\scriptsize  \text{Tx}_2$}
\uput[r](31,36){$\scriptsize  \text{Rx}_2$}

\psline{->}(1.5,38)(24.5,56)
\psline{->}(1.5,42)(24.5,60)

\psline{->}(1.5,68)(24.5,30)

\rput(0,52){\pscircle(0,4){1.5}
\pscircle(0,8){1.5}
\pscircle(0,12){1.5}
\pscircle(0,16){1.5}
\psline{->}(1.5,16)(24.5,12)
\psline{->}(1.5,8)(24.5,4)
\psline{->}(1.5,12)(24.5,8)
\pspolygon(-2,0)(-6,0)(-6,20)(-2,20)
}

\uput[l](-5,62){$\scriptsize  \text{Tx}_1$}
\uput[r](31,62){$\scriptsize  \text{Rx}_1$}

\rput(26,26){
\pscircle(0,4){1.5}
\pscircle(0,8){1.5}
\pscircle(0,12){1.5}
\pscircle(0,16){1.5}
\pspolygon(2,0)(6,0)(6,20)(2,20)
}

\rput(26,52){
\pscircle(0,4){1.5}
\pscircle(0,8){1.5}
\pscircle(0,12){1.5}
\pscircle(0,16){1.5}
\pspolygon(2,0)(6,0)(6,20)(2,20)}


\rput(45,30){\psline[linestyle=dashed,dash=3pt 2pt](0,0)(56,0)
\psline(1,0)(1,21)(7,21)(7,0) \rput(4,-2){\scriptsize 1}
\psline(13,0)(13,14)(19,14)(19,0)\rput(16,-2){\scriptsize 2}

\psline(36,0)(36,7)(42,7)(42,0)\rput(39,-2){\scriptsize 1}
\psline(48,0)(48,28)(54,28)(54,0) \rput(51,-2){\scriptsize 2}

\uput[r](7,21){\scriptsize  $\alpha_{11}$} 
\uput[r](19,14){\scriptsize  $\alpha_{12}$}

\uput[l](36.5,7){\scriptsize  $\alpha_{21}$}
\uput[r](54,28){\scriptsize  $\alpha_{22}$}

\psline[linestyle=dashed,dash=1pt 1pt](1,7)(7,7)
\psline[linestyle=dashed,dash=1pt 1pt](1,14)(7,14)

\psline[linestyle=dashed,dash=1pt 1pt](13,7)(19,7)

\psline[linestyle=dashed,dash=1pt 1pt](48,7)(54,7)
\psline[linestyle=dashed,dash=1pt 1pt](48,14)(54,14)
\psline[linestyle=dashed,dash=1pt 1pt](48,21)(54,21)

\rput(10,36){$\text{Rx}_1$}
\rput(45,36){$\text{Rx}_2$}


}
\end{pspicture}
\end{center} 
\caption{At left is a deterministic interference channel. The more
compact figure at right shows only the signals as observed at the
receivers.} \label{fig:ICrectangles2}
\end{centering}
\end{figure}
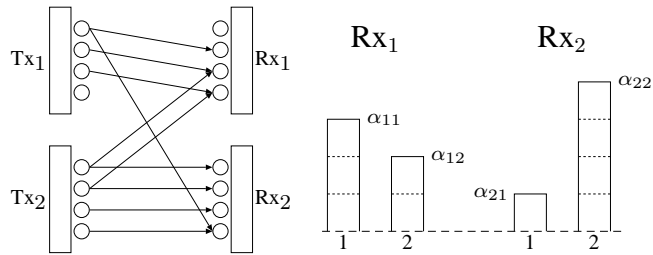

The deterministic interference channel is relatively simple, yet
retains two essential features of the Gaussian interference channel:
the loss of information due to noise, and the superposition of
transmitted signals at each receiver. The modeling of noise can be
understood through the point-to-point channel above. The
superposition of transmitted signals at each receiver is captured by
taking the modulo 2 sum of the incoming signals at each level, as in
the model for the multiple-access channel.

\subsection{Main Results}

To begin, we give the capacity region, $\C$, of our two-user
deterministic interference channel. This region is given by Theorem
1 in \cite{EGC82}, which applies to a larger class of deterministic
interference channels. For our model, the resulting region becomes
 the set of non-negative rates satisfying:\footnote{The
boundaries of the region in \cite{EGC82} is given in terms of
conditional entropies that must be maximized over any product
distribution on the channel inputs. For our model the optimizing
input distribution for each bound is always uniform over the input
alphabet. The given bounds follow.}
{\allowdisplaybreaks
\begin{align}
R_1 &\leq n_{11}\\
R_2 &\leq n_{22}\\
R_1+ R_2 &\leq (n_{11} - n_{12})^+ +\max(n_{22},n_{12})\label{eq:sum1}\\
R_1+ R_2 &\leq (n_{22} - n_{21})^+ + \max(n_{11},n_{21})\label{eq:sum2}\\
R_1+ R_2 &\leq \max(n_{21},(n_{11} - n_{12})^+) \nonumber\\
&\quad\quad+\max(n_{12},(n_{22}-n_{21})^+)\label{eq:sum3}\\
2R_1 + R_2 &\leq \max(n_{11},n_{21})+ (n_{11} - n_{12})^+ \nonumber\\
&\quad\quad+ \max(n_{12},(n_{22}-n_{21})^+) \\
R_1+ 2R_2 &\leq  \max(n_{22},n_{12}) + (n_{22} - n_{21})^+ \nonumber\\
&\quad\quad+ \max(n_{21},(n_{11} - n_{12})^+).
\end{align}
}

Our main result, stated in Theorem~\ref{prop:cne} below is to
completely characterize $\CNE$ for the two-user deterministic
interference channel model. This characterization is in terms of
$\mathcal C$ and a ``box'' ${\mathcal B}$ in $\mathbb R_+^2$ given
by (see Fig.~\ref{fig:box})
\[
{\mathcal B} = \{(R_1,R_2): L_i \leq R_i \leq U_i, \forall i =
1,2\},
\]
where for each user $i=1,2$,
\begin{equation}\label{eq:Li}
L_i = (n_{ii}-n_{ij})^+,
\end{equation}
and
\begin{equation}\label{eq:Ui}
U_i =
\begin{cases}
n_{ii}-\min(L_j,n_{ij}), & \text{ if $n_{ij} \leq n_{ii}$},\\
\min((n_{ij}-L_j)^+,n_{ii}), & \text{ if $n_{ij} > n_{ii}$.}
\end{cases}
\end{equation}

\begin{figure}
\begin{centering}
\setlength{\unitlength}{.085in}%

\begin{picture}(10,11)(-1,-1)
\put(-1,9){\makebox(0,0){\small $R_2$}}
\put(3,2){\framebox(5,4){}}
\put(0,0){\vector(0,1){9}}
\put(0,0){\vector( 1, 0){10}}
\put(5.5,4){\makebox(0,0){${\mathcal B}$}}
\put(9,7){\makebox(0,0){\small $(U_1,U_2)$}}
\put(2.7,7){\makebox(0,0){\small $(L_1,U_2)$}}
\put(2.7,1){\makebox(0,0){\small $(L_1,L_2)$}}
\put(9,1){\makebox(0,0){\small $(U_1,L_2)$}}
\put(10,-1){\makebox(0,0){\small $R_1$}}

\end{picture}
\caption{An example of the box ${\mathcal B}$. The values of the
four
  corner points are indicated in the figure.}\label{fig:box}
\end{centering}
\end{figure}
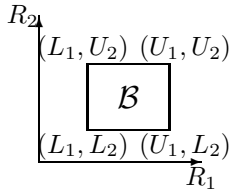

We now state our main result.

\begin{theorem}\label{prop:cne}
$\mathcal C_{NE} = \mathcal C \cap \mathcal B$. 
\end{theorem}

First let us interpret the bounds $L_1,L_2, U_1,U_2$. The number
$L_i$ is the number of levels that user $i$ can transmit above the
interference floor created by user $j$, i.e. the number of levels at
receiver $i$ that cannot see any interference from user $j$. These
are always the most significant bits of user $i$'s transmitted
signal. In the example channel in Figure \ref{fig:ICrectangles2},
these correspond to the top level for transmitter $1$ ($L_1=1$) and
the top $3$ levels for transmitter $2$ ($L_2=3$). The number $U_i$
is the number of levels at receiver $i$ that receive signals from
transmitter $i$ but are free of interference from the top $L_j$
levels from transmitter $j$. In the example channel in
Fig.~\ref{fig:ICrectangles2}, these correspond to the top level at
receiver $1$ ($U_1=1$) and the top three levels at receiver $2$
($U_2=3$).

Intuitively, it is clear that at any $\eta$-NE, user $i$ should have rate
at least $L_i$: these levels are interference-free and user $i$ can
always send information at the maximum rate on these levels. This
will create interference of maximum entropy at a certain subset of
levels at receiver $j$ and render them un-usable for user $j$. The
rate for user $j$ is bounded by the number of remaining levels that
it can use. This is precisely the upper bound $U_j$. What Theorem
\ref{prop:cne} says is that any rate pair in the capacity region $\C$
subject to these natural constraints is in $\CNE$.

To illustrate this result, consider a symmetric interference channel
in which $n_{11}=n_{22}$ and $n_{12}=n_{21}$. Let $\alpha =
n_{ji}/n_{ii}$ be the normalized cross gain. Four examples of
$\mathcal C$ and $\mathcal B$ corresponding to different ranges of
$\alpha$ are shown in Fig.~\ref{fig:intersect}. For
$0<\alpha<\frac{1}{2}$, $\CNE=\mathcal B$ is a single point, which
lies at the symmetric sum-rate point of $\mathcal C$. For
$\frac{1}{2}<\alpha<\frac{2}{3}$, again $\CNE=\mathcal B$. $\CNE$
contains a single efficient point (the symmetric sum-rate point in
$\C$), but now there are additional interior points of $\C$ which
may be achieved as a Nash equilibrium.\footnote{In a slight abuse
 of terminology, we say that points in $\CNE$ can be
``achieved as a NE.''}
 For $\frac{2}{3}<\alpha<1$,
$\CNE$ is the intersection of the simplex formed by the sum-rate
constraint of $\mathcal C$ and $\mathcal B$. In this case, there are
multiple efficient points; in fact, the entire sum-rate face of
$\mathcal C$ is included in $\CNE$. For $1<\alpha<2$, $\mathcal C
\subset \mathcal B$ and so $\CNE = \C$. For $2\leq \alpha$ (not
shown) $\mathcal C = \mathcal B$ and so again $\CNE = \C$. Note that
in all cases, the symmetric rate point is in $\CNE$.

\begin{figure}
\begin{centering}
\psset{unit=0.009in}
\begin{pspicture}(-10,-10)(225,225)
\put(0,130){\vector(0,1){90}} 
\put(0,130){\vector(1,0){90}} 
\psline(20,200)(0,200) 
\psline(70,150)(70,130) 
\psline(20,200)(70,150) 
\put(45,175){\circle*{4}} 
\put(67,197){\vector(-1,-1){19}}
\put(70,203){\makebox(0,0){\small ${\mathcal C}\cap {\mathcal B}$}}
\put(-10,220){\makebox(0,0){\small $R_2$}}
\put(90,120){\makebox(0,0){\small $R_1$}}
\put(40,110){\makebox(0,0){\small $0<\alpha<\frac{1}{2}$}}

\put(120,130){\vector(0,1){90}}
\put(120,130){\vector( 1, 0){90}}
\put(130,140){\framebox(35,35){}}
\put(147,157.5){\makebox(0,0){\small ${\mathcal C}\cap {\mathcal B}$}}
\psline(140,200)(120,210)
\psline(190,150)(200,130)
\psline(140,200)(190,150)
\put(210,120){\makebox(0,0){\small $R_1$}}
\put(110,220){\makebox(0,0){\small $R_2$}} 
\put(160,110){\makebox(0,0){\small $\frac{1}{2}<\alpha<\frac{2}{3}$}}

\put(0,0){\vector(0,1){90}}
\put(0,0){\vector(1,0){90}}
\put(10,10){\framebox(60,60){}}
\put(30,35){\makebox(0,0){\small ${\mathcal C}\cap {\mathcal B}$}}
\psline(20,70)(0,80)
\psline(70,20)(80,0)
\psline(20,70)(70,20)
\put(-10,90){\makebox(0,0){\small $R_2$}}
\put(90,-10){\makebox(0,0){\small $R_1$}}
\put(40,-20){\makebox(0,0){\small $\frac{2}{3}<\alpha<1$}}

\put(120,0){\vector(0,1){90}}
\put(120,0){\vector(1,0){90}}
\put(120,0){\framebox(70,70){}}
\put(150,35){\makebox(0,0){\small ${\mathcal C}\cap {\mathcal B}$}}
\psline(140,70)(120,70)
\psline(190,20)(190,0)
\psline(140,70)(190,20)
\put(110,90){\makebox(0,0){\small $R_2$}}
\put(210,-10){\makebox(0,0){\small $R_1$}}
\put(160,-20){\makebox(0,0){\small $1<\alpha<2$}}

\end{pspicture}
\caption{Examples of $\CNE = \mathcal C \cap \mathcal B$ for a
  symmetric interference channel with normalized cross gain $\alpha$.}
\label{fig:intersect}
\end{centering}
\end{figure}
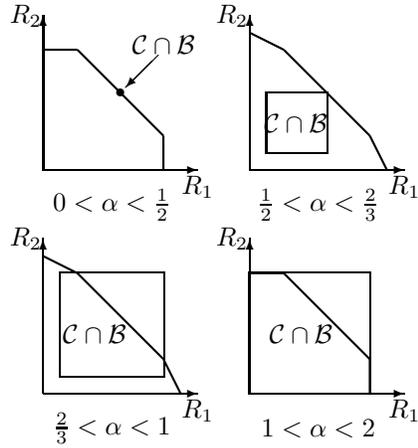

\subsection{Proofs}

To prove Theorem~\ref{prop:cne}, we first show that points outside
of $\mathcal B$ cannot be achievable as a Nash equilibrium,
formalizing the intuition discussed earlier. We will then show that
all points inside $\mathcal C\cap \mathcal B$ can in fact be
achieved.

\subsubsection{Non-equilibrium points}

\begin{lemma}\label{lem:Li}
If $(R_1,R_2) \in \CNE$, then $R_i\geq L_i= (n_{ii}-n_{ij})^+$ for
$i=1,2$.
\end{lemma}
\begin{proof}
User $i$'s $L_i$ highest transmitted levels see no interference from
user $j$'s signal at receiver $i$. Hence, in any
$\epsilon$-game, user $i$ can always achieve a pay-off of $R_i =
L_i$ (with zero probability of error) by sending $L_i$ uncoded bits,
one on each of these levels, independent of user $j$'s strategy.  Thus
for $(R_1,R_2)$ to be obtained as an $\eta$-NE of an $\epsilon$-game
for any $\epsilon >0$, it must be
that $R_{i} > L_{i}-\eta $ for all $i$; otherwise, user $i$ could deviate
using the above strategy and improve his pay-off by $\eta$.
If  $(R_1,R_2) \in \CNE$, then this inequality should hold
for all $\eta>0$. Taking the limit as $\eta\rightarrow 0$, the result follows.
\end{proof}

\begin{lemma}\label{lem:Ui}
If $(R_1,R_2) \in \CNE$, then for all $i= 1,2$, $R_i \leq U_i$,
where $U_i$ is given in (\ref{eq:Ui}).
\end{lemma}
\begin{IEEEproof}
Suppose $(R_1,R_2) \in \CNE$.  Without loss of generality, let us
focus on user $1$ to show that $R_1 \le U_1$.

If $n_{12}-L_2 \ge n_{11}$, then $U_1 = n_{11}$ and clearly $R_1 \le
n_{11}$, so there is nothing to prove. So in the following we can 
assume that
\beq \label{eq:cond} n_{12} - L_2 < n_{11}.\eeq

Fix an arbitrary $\eta > 0$. Given a sufficiently small $\eps > 0$, there exist a
$(1-\eps)$-reliable strategy pair $(s^*_1, s^*_2)$ achieving the
rate pair $(R_1,R_2)$ that is also a $\eta$-NE. As remarked in Section \ref{sec:formulation}, we can assume that in this nominal strategy pair, both users use a common block length $N$.  Applying Fano's inequality
to user $1$ for the average bit error probability (see for example Theorem 4.3.2 in \cite{Gal08}), we get the bound, for any block $k$:
\beq \label{eq:fano}
R_1 \le \frac{I(m_1; \vy_1|\omega_1)}{N} + \delta \eeq
where $\delta$
depends on $\eps$ and goes to zero as $\eps$ goes to zero. Here, $\omega_1$
denotes any randomness in $\vx_1$, which recall is 
known at receiver 1. Note that
we drop the block indices of the message and the signals to simplify
notation. Now,
\begin{eqnarray*} & & \frac{1}{N} I(m_1; \vy_1|\omega_1)\\
& \le & \frac{1}{N} I(\vx_1;\vy_1|\omega_1)\\
& \le &  \frac{1}{N} I(\vx_1;\vy_1)\\
& =  & \frac{1}{N} \left [H(\vy_1) - H(\vs_2) \right]\\
& \le & \max(n_{11},n_{12}) - \frac{H(\vs_2)}{N} \end{eqnarray*} where $\vs_2$
is the signal from user $2$ that is visible at receiver $1$. Here, the second
inequality follows since $\omega_1 - \vx_1 -\vy_1$ forms a Markov chain.
Combining this with
the above inequality, we get:
\beq \label{eq:upper} R_1 \le \max(n_{11},n_{12}) - \frac{1}{N}
H(\vs_2) + \delta. \eeq

We now seek a bound on $H(\vs_2)$.  Applying Fano's inequality to
user $2$, we get
\begin{eqnarray*} R_2   &\le & \frac{1}{N} I(m_2; \vy_2|\omega_2) + \delta\\
& \le & \frac{1}{N} I(\vx_2;\vy_2) +\delta\\
& = & \frac{1}{N} I(\vu_2,\vv_2; \vy_2) +\delta
\end{eqnarray*}
where $\vu_2$ is the
part of user 2's transmitted signal $\vx_2$ which is on the top
$\min(L_2, n_{12})$ levels and $\vv_2$ is the rest. The significance
of $\vu_2$ is that it is received without interference at receiver
$2$ {\em and} is visible at receiver $1$ (i.e., part of $\vs_2$).
Correspondingly, split user 2's received signal $\vy_2 =
(\vu_2,\tvy_2)$. Now

\begin{eqnarray*}
 & & I(\vu_2, \vv_2; \vy_2)\\
& = & I(\vu_2;\vy_2) + I(\vv_2;\vy_2|\vu_2)\\
& = & H(\vu_2) + I(\vv_2;\tvy_2|\vu_2)\\
& \le & H(\vu_2) + I(\vv_2;\tvy_2) \end{eqnarray*}
 since $\vu_2 -
\vv_2 - \tvy_2$ forms a Markov chain.

Combining this with the previous equation, we get:
\beq \label{eq:upper2} R_2   \le \frac{1}{N} H(\vu_2) + \frac{1}{N}
I(\vv_2;\tvy_2) + \delta. \eeq

Let us now consider an alternative strategy $s'_2$ for user 2. This
encoding strategy has two independent sub-codes. The first sub-code
transmits uncoded bits on the top $\min(L_2, n_{12})$ levels,
achieving a rate of $\min(L_2, n_{12})$ bits per symbol time with
zero error. The second sub-code transmits on the remaining $n_{22}-
\min(L_2,n_{12})$ levels. It codes over $K$ blocks of length $N$
each. Each codeword in this code has $K$ components, each spanning
$N$ symbol times, for a total length of $KN$ symbol times. Each
codeword is chosen randomly, with i.i.d. $N$-length components and
each component chosen from the distribution of user 2's transmit
signal $\vv_2$ under the original encoding strategy $s_2^*$. Note
that since user 1's strategy $s_1^*$ codes only within blocks of
length $N$ and sends independent message across different blocks,
the interference from user 1 is i.i.d. across such blocks. User 2
thus faces a memoryless channel from block to block. Standard random
coding arguments apply and one can show that for any $\delta_1 > 0$,
there exists a large enough $K$ such that strategy $s'_2$ can
achieve a rate of $I(\vv_2; \tvy_2) - \delta_1$ bits per block and
with a probability of error of less than $\eps$. Thus, strategy
$s_2'$ achieves a total rate $R_2'$ bits per symbol time reliably,
where
\beq \label{eq:prime} R_2' =  \min(L_2, n_{12}) + \frac{1}{N}
[I(\vv_2;\tvy_2) - \delta_1]. \eeq

By definition of $\eta$-NE, strategy $s_2'$ cannot perform
much better than $s_2^*$, i.e., $R_2 + \eta  \ge R'_2$. Combining (\ref{eq:prime})
and (\ref{eq:upper2}), we now have:
\beq \frac{1}{N} H(\vu_2) \ge \min(L_2, n_{12}) - \delta_1/N -\delta - \eta.
\eeq

Essentially, we have shown that user 2 under strategy $s_2^*$ must
be transmitting information at maximum entropy on these $\min\{L_2,
n_{12})$ levels by virtue of the fact that it forms a $\eta$-NE.
Substituting this into (\ref{eq:upper}) and observing that $H(\vs_2)
\ge H(\vu_2)$, we get:

\beq  R_1 \le \max(n_{11},n_{12}) -  \min(L_2, n_{12}) +\delta_1/N +
2 \delta +\eta. \eeq

Under the condition (\ref{eq:cond}), one can readily verify that
$U_1 = \max(n_{11},n_{12}) -  \min(L_2, n_{12})$.  Since $\eta$, $\delta$
and $\delta_1$ can be chosen arbitrarily small, we have shown that
$R_1 \le U_1$. The proof is complete.

\end{IEEEproof}

\subsubsection{Achievable Nash Equilibria}

From Lemmas~\ref{lem:Li} and \ref{lem:Ui}, it follows that
${\mathcal C}_{NE} \subseteq \mathcal C \cap \mathcal B$. 
In this section, we
show that these two sets are in fact equal,
proving Theorem~\ref{prop:cne}.
To do this we consider a modification of the class of
Han-Kobayashi strategies presented in \cite{HK81}. In these strategies,
each user splits the transmitted information into two parts:
private information to be decoded at only their own receiver and
common information that can be decoded at both receivers. In \cite{BT08}
it is shown that a particular class of these strategies can achieve
any point in the capacity region of the deterministic channel. The
modification we make to these schemes is to allow each transmitter
to include extra random bits in their common message. Next, we give some
preliminary definitions related to Han-Kobayashi schemes and then
formally define this class of strategies.

For a given deterministic interference channel, let ${\mathcal X}_i$
denote the input alphabet of user $i$, i.e.~this is the set of 
$\max(n_{ii},n_{ji})$. We decompose this set as the direct product
${\mathcal X}_{ic} \times {\mathcal X}_{ip}$, so that for any
$\vx_i \in \mathcal X_i$ can be written as
$\vx_i = (\vx_{ip},\vx_{ic})$,
where $\vx_{ip}$ denotes the $(n_{ii}-n_{ji})^+$ least significant
levels of $\vx_{i}$, and
$\vx_{ic}$ consists of the $n_{ji}$ most significant levels.
The significance of this decomposition is that the $\vx_{ip}$ consists of
{\it private levels} for user $i$ which are visible only at his receiver, while
$\vx_{ic}$ consists of {\it common levels} that are visible at 
receiver $j$.

We define a {\it randomized Han-Kobayashi} scheme
for a given block-length $N$ to be a scheme in which each each user $i$
separates the message set $\{1,\ldots, 2^{B_i}\}$ into the direct
product of a private message ${\mathcal M}_{ip}$ set containing
$2^{NR_{ip}}$ messages and a common message set ${\mathcal M}_{ic}$
containing $2^{NR_{ic}}$ messages, where
$NR_{ip} + NR_{ic} = B_i$. Additionally, each user $i$ is allowed 
to have a {\it random common message} set $\Omega_i$ 
consisting of $2^{NR_{ir}}$ equally likely
codewords; these can be thought of as $NR_{ir}$ random bits that 
the transmitter generates using the common randomness, which is 
shared with the corresponding
receiver. The message sets are then encoded using a superposition code
as follows.
First the transmitter encodes the common and common random message
via a map $f_{ic}: {\mathcal M}_{ic}\times\Omega_i \mapsto
{\mathcal X}_{ic}^N$, where the codebook is generated 
using an i.i.d.~uniform distribution over the common levels.
Next, the
transmitter encodes the private message via a map
$f_{ip}: {\mathcal M}_{ip}\times {\mathcal M}_{ic} \times \Omega_i
\mapsto {\mathcal X}_{ip}^N$, where for each common codeword $\vx_{ic}$, 
a different private codebook is generated using an i.i.d.~uniform 
distribution over the private levels.
Here, the common codeword $\vx_{ic}$ can
be viewed as defining the cloud center and the 
the private codeword can be
viewed as defining the cloud points. 
Transmitter $i$ then sends
the superposition of these two codewords.
In the special case where $R_{1r} =R_{2r} =0$, we 
refer to the resulting scheme 
as a {\it non-randomized Han-Kobayashi} scheme.

We call a randomized Han-Kobayashi scheme {\it $(1-\epsilon)$-reliable}
if each user $i$ can decode their own private and common messages with
an average probability of bit error
no greater than $\epsilon$. \footnote{This is slight 
strengthening of the previous definition of a reliable 
strategy, which only
required the overall average bit error probability to be no greater
than $\epsilon$. Hence, an agent's pay-off in a
$\epsilon$-game under a $(1-\epsilon)$-reliable Han-Kobayashi scheme
is again their rate.}  

Next we specify an achievable rate region for this class of schemes. This
characterization is in terms of {\it modified
MAC regions} for each of the two receivers. Specifically, 
the modified MAC region at receiver $i$, $\mathcal R_{i}^m$, 
is the set of rates 
$(R_{ic},R_{ir},R_{ip},R_{jc},R_{jr},R_{jc})$ that satisfy:
\begin{equation}\label{eq:MACi}
\begin{split}
R_{ic} +R_{ip} + R_{jc} +R_{jr} &\leq \max(n_{ii},n_{ij}) \\
R_{ip}+ R_{jc} + R_{jr} &  \leq \max(n_{ij},(n_{ii}-n_{ji})^+)  \\
R_{ip} & \leq (n_{ii} -n_{ji})^+\\
R_{ic} +R_{ip} & \leq n_{ii}
\end{split}
\end{equation}

The modified MAC region is derived by 
considering the MAC channel at receiver $i$ consisting 
of three transmitters one corresponding to user $i$'s 
own common message, one corresponding to user $i$'s own private
messages, and one corresponding to the combination user $j$'s 
common and common random messages. Here, user $i$'s own common random 
signal can be ignored
since it is known at receiver $i$ and so can be removed. Likewise,
user $j$'s private message can be ignored since it is not received at
receiver $i$. The capacity region of this three user MAC will have seven
constrains including the first four given in \eqref{eq:MACi}.
The modification to this is that we drop the remaining three
constraints by following similar arguments as in \cite{CMGE08}.
First, recall that 
a constraint for a MAC region that involves the 
rates $\{R_{i}: i\in \mathcal M\}$ for
some subset of the users $\mathcal M$ corresponds to a bound on an
error event where the message for each user in $\mathcal M$ is in
error and all other messages are correct \cite{Gal85}. The missing bounds
correspond to error events that we ignore for one of the
following two reasons. First, due to the use of superposition 
coding we can ignore constraints which correspond to making an 
error in the common message but not the private message. 
Second, from the point-of-view of user
$i$, we can ignore the constraint that corresponds to an error in only
user $j$'s signal.

From the above discussion and following similar arguments 
as in \cite{BT08},
we then have the following characterization of the rate-tuples
that be achieved with this class of randomized Han-Kobayashi schemes.

\begin{lemma}\label{lem:randomized}
$\mathcal R_{RHK} = \mathcal R_{1}^m \cap \mathcal R_{2}^m$ is an
achievable region for randomized Han-Kobayashi schemes.
\end{lemma}

If a rate-tuple is in $\mathcal R_{RHK}$, receiver $i$ may not be able
to reliably decode user $j$'s common and common random messages. In 
particular, this must be true if $R_{jc} + R_{jr} > n_{ij}$. However,
if $R_{ic} + R_{ip} \geq L_i$  for $i=1,2$, then the next lemma shows
that this will be possible.

\begin{lemma}\label{lem:decode_common}
Any rate-tuple $(R_{1c}, R_{1r}, R_{1p},R_{2c},R_{2r},R_{2p})$ in the interior of 
$\R_{RHK}$ with $R_{ic} + R_{ip} \geq L_i$ for $i = 1,2$ 
can be achieved by a randomized Han-Kobayashi scheme in 
which each user $i$ decodes user $j$'s common and common random
message (with arbitrarily small probability of error).
\end{lemma}

\begin{proof}
First note that if a rate tuple is in the interior of $\R_{RHK}$ with 
$R_{ic} + R_{ip} \geq L_i = (n_{ii} - n_{ij})^+$, then
from the first constraint in (\ref{eq:MACi}) for receiver
$i$, it follows that
\begin{equation}\label{eq:fifth}
R_{jc} + R_{jr} < n_{ij}.
\end{equation}
Now consider a randomized Han-Kobayashi scheme which achieves this 
rate tuple. After user $i$ decodes his own private and common
messages, he will have a clean view of user $j$'s common
message. Moreover, from the above constraint, the rate of this message
is less than the capacity of the channel from user $j$ to receiver $i$
and so there must exist a randomized Han-Kobayashi scheme in which
user $i$ can also decode user $j$'s common messages with arbitrary 
reliability.\footnote{The 
  Han-Kobayashi scheme under consideration may suffice, however if
  this scheme does not provide a low enough probability of error,
  then a new scheme that achieves the same rate-tuple with the desired
  probability of error can be found (perhaps by using a longer
  blocklength).}
\end{proof}

A given rate tuple $\mathbf {R} =(R_{1c},R_{1r},R_{1p},R_{2c},R_{2r},R_{2p})$ 
is defined to be {\it self-saturated} at receiver $i$ if 
$\mathbf{R} \in \R_i^m$ at receiver $i$, 
but any other choice of $R_{ic}$, $R_{ir}$ and 
$R_{ip}$ with a larger value of
$R_{ic} +R_{ip}$ (keeping all other rates fixed)
will result in a rate-tuple that is not in $\R_{i}^m$.
Clearly, a self-saturated rate-tuple must lie on the boundary of the
modified MAC region at receiver $i$. Additionally, the constraints 
in \eqref{eq:MACi} that
are tight at this point must involve both $R_{ip}$ and $R_{ic}$.
If a rate-pair is self-saturated and $R_{ip} + R_{ic} \geq L_i$, then
it will be useful to think about this in the context of a second MAC
region at receiver $i$ in which there are only two users, one
corresponding to user $i$'s entire message (at rate $R_i =
R_{ic}+R_{ip}$) and a second that again corresponds to the common and
common random messages sent by user $j$ (at rate $R_{jc} + R_{jr}$).
This MAC region is given by
\begin{equation}\label{eq:MACi2}
\begin{split}
R_{i} + R_{jc} +R_{jr} &\leq \max(n_{ii},n_{ij}) \\
R_{i} & \leq n_{ii}\\
R_{jc} +R_{jr} &\leq n_{ij}.
\end{split}
\end{equation}
It can seen that if a rate tuple is self-saturated for user $i$ and
satisfy $R_{ic} + R_{ip} \geq L_i$, then
the rates $R_{i}$ and $R_{jc} + R_{jr}$ must be in (\ref{eq:MACi2}) and if
$R_i$ is increased by any amount this will no longer be true.
We use this to show that if a user is self-saturated
and his rate is greater than $L_i$,
then he can
not deviate and improve his pay-off. The key idea here is that if a user could
improve, then he will violate one of the MAC constraints in
(\ref{eq:MACi2}). This cannot be possible since after deviating and
decoding his own message, he should still be able to decode the 
other user's message. The next lemma formalizes this argument.

\begin{lemma}\label{lem:int}
If a rate tuple
$(R_{1c},R_{1r},R_{1p},R_{2c},R_{2r},R_{2p})$
is self-saturated with $R_{ic} + R_{ip} \geq L_i$
for both both receivers $i$, then
$(R_{1p}+R_{1c},R_{2p}+R_{2c})\in \CNE$.
\end{lemma}

\begin{IEEEproof}
Given a rate tuple  $\vR= (R_{1c},R_{1r},R_{1p},R_{2c},R_{2r},R_{2p})$
that is self-saturated at both receivers and with $R_{ic} + R_{ip} \geq L_i$, it
follows from Lemmas~\ref{lem:randomized} and \ref{lem:decode_common} 
that for any $\eta>0$,
and any $\epsilon>0$, there exists a 
randomized Han-Kobayashi scheme achieving rates
$(R_{1c}-\eta/6,R_{1r}-\eta/6,R_{1p}-\eta/6,R_{2c}-\eta/6,R_{2r}-\eta/6,R_{2p}-\eta/6)$
for which each receiver can decode both his own common and private
messages as well as the other users common and common random messages
with probability of error less than $\epsilon$.

Next we argue that for $\epsilon$ small enough such a pair of 
strategies must be a $\eta$-NE of a $\epsilon$-game.
First note that under these nominal strategies each
user $i$ will receive a pay-off of $R_{ic} + R_{ip} -\eta/3$. Assume
that these strategies are not 
an equilibrium, and without loss of generality suppose that user $1$ can
deviate and improve his performance by at least $\eta$. After
deviating, the rates for the MAC
region at user 1 in (\ref{eq:MACi2}) are given by
$\tilde{\vR} = (\tilde{R}_{1}, R_{2c} + R_{2r}-\eta/3)$, where
$\tilde{R}_1 \geq R_1 + 2\eta/3$. Since the rate tuple
$\vR$ is self-saturated, it follows that after this deviation,
the rates $\tilde{\vR}$
must violate either the first or second constraint in (\ref{eq:MACi2})
for $i=1$ by at least $\eta/3$.

Suppose that user $1$ deviates to a blocklength $N_1$ strategy. Then
using Fano's inequality (as in (\ref{eq:fano})) for the average bit
error probability and following the usual converse for a MAC channel,
it must be that
\begin{equation}\label{eq:tr1}
\begin{split}
\tilde{R}_1 \leq \frac{I(\vx_1;\vy_1|\vx_{2c})}{N_1} + \delta
\end{split}
\end{equation}
where $\delta$ goes to zero as the average bit error probability
$\epsilon$ does. In particular, choosing $\epsilon$ small enough so
that  $\delta$ is less than $\eta/3$, then (\ref{eq:tr1}) implies that
\[
\tilde{R}_1 < n_{11} + \eta/3.
\]
Hence, the second constraint in (\ref{eq:MACi2}) can not be violated.

Likewise, since in the nominal strategy for user $1$, he
was able to decode user $j$'s common and common random signals,
it follows that
\begin{equation}\label{eq:tr2}
R_{2c} + R_{2p} \leq \frac{I(\vx_{2c};\vy_1)}{N_i'} + \delta'
\end{equation}
where $N_i'$ denotes the block length used in the nominal strategy.
Combining (\ref{eq:tr1}) and (\ref{eq:tr2}) and
choosing $\epsilon$ small enough so that $\delta + \delta' <
\eta/3$, we have\footnote{Note, as discussed in
  Section~\ref{sec:formulation}, we need to replace (\ref{eq:tr1}) and (\ref{eq:tr2})
  with the corresponding expressions over super-blocks whose length is
  the least common multiple of $N_i$ and
  $N_i'$ so that that both equations are over the same block-length.}
\begin{equation}
\tilde{R_1} + R_{2c} + R_{2p} \leq \max(n_{ii},n_{ij}) + \eta/3
\end{equation}
which shows that the first constraint in (\ref{eq:MACi2}) can not
be violated.

Therefore, such a deviation cannot exist and the nominal
strategy must be a $\eta$-NE for small enough $\epsilon$.
Taking the limit as $\eta\rightarrow 0$,
it follows that the desired rates must lie in $\CNE$.
\end{IEEEproof}

To prove that $\CNE = \C \cap \B$, we will show for that any 
rate-point $(R_1,R_2) \in \C\cap\B$, there exists a 
feasible rate tuple in $\R_{RHK}$ with $R_{i} = R_{ic} +R_{ip}$ for $i=1,2$ and 
that is self-saturated at both receivers. The desired result then
follows directly from Lemma~\ref{lem:int}. As a first step toward
doing this, we define a class of non-randomized Han-Kobayashi rates
$(R_{1c},R_{1p},R_{2c},R_{2p})$ at which each transmitter is fully
utilizing its ``interference-free'' levels, i.e., the $L_i = (n_{ii}
-n_{ij})^+$ most significant levels at transmitter $i$. Depending on
the channel some of these levels may be common and some may be
private. Specifically, there are
\[
a_{i} = (n_{ii} - n_{ji} - n_{ij})^+
\]
private interference free levels at user $i$ and
\[
b_{i} = (n_{ii} - \max(n_{ii}-n_{ji},n_{ij}))^+
\]
common interference free levels at user $i$. An example of these is
shown in Fig.~\ref{fig:int_free}.
We say that $(R_{1c},R_{1p},R_{2c},R_{2p})$ {\it fully utilizes the
  interference free levels} for user $i$ if
\begin{align}
R_{ip} &\geq a_i \label{eq:ai}\\
R_{ic} &\geq b_i  \label{eq:bi}.
\end{align}

\begin{figure}
\begin{centering}
\psset{unit=.85mm,linewidth=.5pt,arrowlength=1.2,
arrowinset=0,labelsep=2pt}
\begin{center}
\begin{pspicture}(0,-1)(57,40)

\rput(0,0){
\psline[linestyle=dashed,dash=3pt 2pt](0,0)(21,0)
\psline[linestyle=dashed,dash=3pt 2pt](35,0)(56,0)
\psline(1,0)(1,21)(7,21)(7,0) \rput(4,-2){\scriptsize 1}
\psline(13,0)(13,14)(19,14)(19,0)\rput(16,-2){\scriptsize 2}

\psline(36,0)(36,7)(42,7)(42,0)\rput(39,-2){\scriptsize 1}
\psline(48,0)(48,28)(54,28)(54,0) \rput(51,-2){\scriptsize 2}

\uput[r](7,21){\scriptsize  $n_{11}$} 
\uput[r](19,14){\scriptsize  $n_{12}$}

\uput[l](36.5,7){\scriptsize  $n_{21}$}
\uput[r](54,28){\scriptsize  $n_{22}$}

\psline[linestyle=dashed,dash=1pt 1pt](1,7)(7,7)
\psline[linestyle=dashed,dash=1pt 1pt](1,14)(7,14)

\psline[linestyle=dashed,dash=1pt 1pt](13,7)(19,7)

\psline[linestyle=dashed,dash=1pt 1pt](48,7)(54,7)
\psline[linestyle=dashed,dash=1pt 1pt](48,14)(54,14)
\psline[linestyle=dashed,dash=1pt 1pt](48,21)(54,21)

\rput(10,36){$\text{Rx}_1$}
\rput(45,36){$\text{Rx}_2$}

\rput(4,17.5){\scriptsize $b$}

\rput(51,24.5){\scriptsize $b$}
\rput(51,17.5){\scriptsize $b$}
\rput(51,10.5){\scriptsize $a$}

}
\end{pspicture}
\end{center} 
\caption{This figure shows the interference free levels for each
  transmitter in the interference channel from
  Fig.~\ref{fig:ICrectangles2}. The levels for user $i$ are indicated
  at that user's receiver by either a ``a'' or a ``b''. The levels
  labeled with a ``a'' correspond to ``private levels,'' which are
  accounted for in (\ref{eq:ai}). The levels labeled with a ``b''
  correspond to common levels, which are accounted for in 
  (\ref{eq:bi}).}\label{fig:int_free}
\end{centering}
\end{figure}
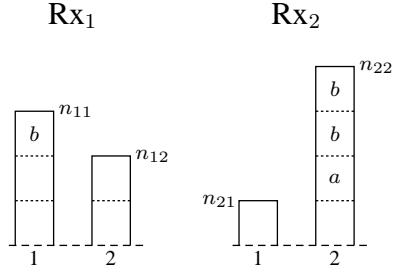

\begin{lemma}\label{lem:full}
If $(R_{1c},R_{1p},R_{2c},R_{2p})$ {\it fully utilizes the
  interference free levels} for each user $i$, then
there exists random common rates $R_{1r},R_{2r} \ge 0$ such that
$(R_{1c},R_{1r},R_{1p},R_{2c},R_{2r},R_{2p})$ is self-saturated at
both receivers.
\end{lemma}

\begin{IEEEproof}
The intuition is as follows. If $R_{ip} \ge a_i$ and $R_{ic} \ge
b_i$, then all the interference-free levels are saturated at
receiver $i$ (i.e., at maximum entropy). The remaining $n_{ij}$ levels
are all reachable by the common signal from user $j$. By putting
sufficient number of random bits on that common signal, these
$n_{ij}$ levels can be fully saturated as well.

More rigorously, we will show that one can always increase $R_{2r}$
such that the overall sum rate constraint (the first constraint in
(\ref{eq:MACi})) is tight, so that receiver $1$ is
saturated. Suppose no such choice of $R_{2r}$ exists. Then, it must
be that $R_{2r}$ cannot be further increased because the
second constraint in (\ref{eq:MACi}) for receiver $i$
is tight. However, if this constraint is tight, 
then since $R_{1c} \geq b_i$,
it can be seen that the first constraint at receiver $i$ 
must also be tight.
\end{IEEEproof}

As an example of the construction used in the proof of
Lemma~\ref{lem:full} consider a symmetric channel with $n_{11} =
n_{22} = 3$ and $n_{12} = n_{21} = 2$ as shown in Fig.~\ref{fig:lem6_1}. 
Each user has $1$
interference free level with $a_i = 0$ and $b_i = 1$. Thus the
non-randomized Han-Kobayashi rates given by $R_{1c} = R_{2c} = 1$
and $R_{1p} = R_{2p} = 0$ fully utilize the interference free levels
at each receiver, but is not self-saturating since, each transmitter
could increase $R_{ip}$ by one. However, if each transmitter
sets $R_{ir}=1$, the resulting randomized rates will be
self-saturated as shown in Fig.~\ref{fig:lem6_2}

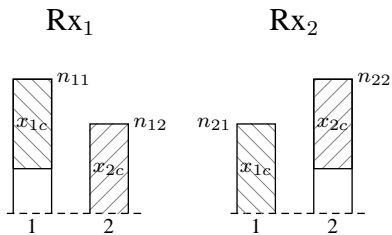
\begin{figure}
\begin{centering}
\psset{unit=.85mm,linewidth=.5pt,arrowlength=1.2,
arrowinset=0,labelsep=2pt}
\begin{center}
\begin{pspicture}(0,-1)(57,40)

\rput(0,0){
\psline[linestyle=dashed,dash=3pt 2pt](0,0)(21,0)
\psline[linestyle=dashed,dash=3pt 2pt](35,0)(56,0)
\psline(1,0)(1,21)(7,21)(7,0) \rput(4,-2){\scriptsize 1}
\psline[fillstyle=hlines,hatchcolor=gray,hatchwidth=.3pt](13,0)(13,14)(19,14)(19,0)
\rput(16,-2){\scriptsize 2}
\psline[fillstyle=vlines,hatchcolor=gray,hatchwidth=.3pt](36,0)(36,14)(42,14)(42,0)
\rput(39,-2){\scriptsize 1}
\psline(48,0)(48,21)(54,21)(54,0) \rput(51,-2){\scriptsize 2}

\uput[r](7,21){\scriptsize  $n_{11}$} 
\uput[r](19,14){\scriptsize  $n_{12}$}

\uput[l](36,14){\scriptsize  $n_{21}$}
\uput[r](54,21){\scriptsize  $n_{22}$}


\psline[fillstyle=vlines,hatchcolor=gray,hatchwidth=.2pt](1,7)(1,21)(7,21)(7,7)
\psline[fillstyle=hlines,hatchcolor=gray,hatchwidth=.2pt](48,7)(48,21)(54,21)(54,7)

\psline(1,7)(7,7)
\psline(48,7)(54,7)



\rput(10,30){$\text{Rx}_1$}
\rput(45,30){$\text{Rx}_2$}

\rput(4,14){\scriptsize $x_{1c}$}
\rput(16,7){\scriptsize $x_{2c}$}
\rput(39,7){\scriptsize $x_{1c}$}
\rput(51,14){\scriptsize $x_{2c}$}

}
\end{pspicture}
\end{center} 
\caption{An example of a non-randomized Han-Kobayashi scheme for a 
symmetric channel with $n_{ii} =3$ and $n_{ij} = 2$. Here each user is
using the rate split $R_{ic} =1$ and $R_{ip} =0$, which fully utilizes
the 1 interference free-level at each transmitter. This rate split
is not self-saturating at either transmitter.}\label{fig:lem6_1}
\end{centering}
\end{figure}

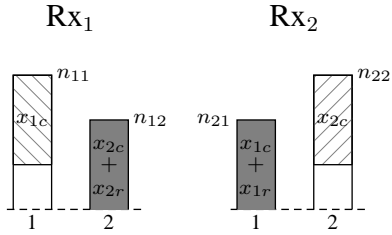
\begin{figure}
\begin{centering}
\psset{unit=.85mm,linewidth=.5pt,arrowlength=1.2,
arrowinset=0,labelsep=2pt}
\begin{center}
\begin{pspicture}(0,-1)(57,40)

\rput(0,0){
\psline[linestyle=dashed,dash=3pt 2pt](0,0)(21,0)
\psline[linestyle=dashed,dash=3pt 2pt](35,0)(56,0)
\psline(1,0)(1,21)(7,21)(7,0) \rput(4,-2){\scriptsize 1}
\psline[fillstyle=solid,fillcolor=gray](13,0)(13,14)(19,14)(19,0)
\rput(16,-2){\scriptsize 2}
\psline[fillstyle=solid,fillcolor=gray,hatchwidth=.3pt](36,0)(36,14)(42,14)(42,0)
\rput(39,-2){\scriptsize 1}
\psline(48,0)(48,21)(54,21)(54,0) \rput(51,-2){\scriptsize 2}

\uput[r](7,21){\scriptsize  $n_{11}$} 
\uput[r](19,14){\scriptsize  $n_{12}$}

\uput[l](36,14){\scriptsize  $n_{21}$}
\uput[r](54,21){\scriptsize  $n_{22}$}


\psline[fillstyle=vlines,hatchcolor=gray,hatchwidth=.3pt](1,7)(1,21)(7,21)(7,7)
\psline[fillstyle=hlines,hatchcolor=gray,hatchwidth= .3pt](48,7)(48,21)(54,21)(54,7)

\psline(1,7)(7,7)
\psline(48,7)(54,7)



\rput(10,30){$\text{Rx}_1$}
\rput(45,30){$\text{Rx}_2$}

\rput(4,14){\scriptsize $x_{1c}$}
\rput(16,10){\scriptsize $x_{2c}$}
\rput(16,7){\scriptsize $+$}
\rput(16,3){\scriptsize $x_{2r}$}
\rput(39,10){\scriptsize $x_{1c}$}
\rput(39,7){\scriptsize $+$}
\rput(39,3){\scriptsize $x_{1r}$}
\rput(51,14){\scriptsize $x_{2c}$}

}
\end{pspicture}
\end{center} 
\caption{A randomized Han-Kobayashi scheme which achieves the same
  rates as the non-randomized scheme in 
  Fig.~\ref{fig:lem6_2} but is self-saturated. Here we do not show
  $R_{ir}$ at receiver $i$ since this signal can be removed from the
  assumed common randomness.}\label{fig:lem6_2}
\end{centering}
\end{figure}

It follows from \cite{BT08} that for any rate pair in $\C\cap \B$, 
there exists a non-randomized Han-Kobayashi rate-split that satisfies 
Lemma~\ref{lem:randomized}. If
these Han-Kobayashi rates fully utilize the interference-free 
levels then we
are done. Unfortunately, not all non-randomized Han-Kobayashi rates
fully utilize the interference-free levels. For example consider the
symmetric channel
in the previous paragraph. An alternative non-randomized
Han-Kobayashi rate-split is given by $R_{1c}=R_{2c} = 0$ and $R_{1p}=
R_{2p} = 1$ (see Fig.~\ref{fig:lem7}). These rates do not fully utilize the interference-free
levels and cannot be made into an equilibrium by simply increasing
the users' common random rates. Though this set of rates do not fully
utilize the interference-free levels, as the previous example
illustrates, there is another set of non-randomized rates that do. The
next lemma generalizes this example.

\begin{figure}
\begin{centering}
\psset{unit=.85mm,linewidth=.5pt,arrowlength=1.2,
arrowinset=0,labelsep=2pt}
\begin{center}
\begin{pspicture}(0,-1)(57,40)

\rput(0,0){
\psline[linestyle=dashed,dash=3pt 2pt](0,0)(21,0)
\psline[linestyle=dashed,dash=3pt 2pt](35,0)(56,0)
\psline(1,0)(1,21)(7,21)(7,0) \rput(4,-2){\scriptsize 1}
\psline(13,0)(13,14)(19,14)(19,0)
\rput(16,-2){\scriptsize 2}
\psline(36,0)(36,14)(42,14)(42,0)
\rput(39,-2){\scriptsize 1}
\psline(48,0)(48,21)(54,21)(54,0) \rput(51,-2){\scriptsize 2}

\uput[r](7,21){\scriptsize  $n_{11}$} 
\uput[r](19,14){\scriptsize  $n_{12}$}

\uput[l](36,14){\scriptsize  $n_{21}$}
\uput[r](54,21){\scriptsize  $n_{22}$}


\psline[fillstyle=solid,fillcolor=gray](1,0)(1,7)(7,7)(7,0)
\psline[fillstyle=solid,fillcolor=gray](48,0)(48,7)(54,7)(54,0)

\psline(1,7)(7,7)
\psline(48,7)(54,7)



\rput(10,30){$\text{Rx}_1$}
\rput(45,30){$\text{Rx}_2$}

\rput(4,3){\scriptsize $x_{1p}$}
\rput(51,3){\scriptsize $x_{2p}$}

}
\end{pspicture}
\end{center} 
\caption{An example of an alternative non-randomized Han-Kobayashi
  rate-split that achieves the same rates as the scheme in
  Fig.~\ref{fig:lem6_1} but does not fully utilize the 
interference-free levels.}\label{fig:lem7}
\end{centering}
\end{figure}
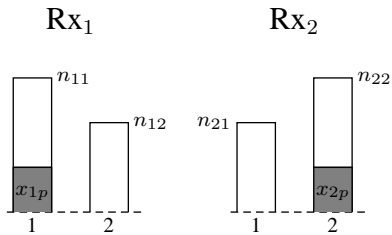

\begin{lemma}\label{lem:fullutil}
Given any point $(R_1,R_2) \in \C \cap \B$, then there exists a
non-randomized Han-Kobayashi rate split that fully utilizes the
interference free levels at each transmitter.
\end{lemma}

\begin{IEEEproof}
To prove this lemma we begin with an arbitrary non-randomized
Han-Kobayashi rate-split that satisfies Lemma~\ref{lem:randomized}
and show that this can always be transformed
into one that fully utilizes the interference free levels at each
transmitter. A key point here is that
$a_{i} + b_{i} = L_i$, so that for any point in $\B$ there will
always be sufficient amount of ``rate'' available to meet the
constraints in (\ref{eq:ai}) and (\ref{eq:bi}).

First we show that if $R_{ip} < a_{i}$ for either user $i$,
then we can always increase $R_{ip}$ and decrease $R_{ic}$ by the 
same amount until
$R_{ip} = a_{i}$. 
The only way such a transformation could not be done is
if the second constraint in (\ref{eq:MACi}) at receiver $i$
prevented it.  But by combining (\ref{eq:fifth}) at receiver $j$ and 
the second constraint at receiver $i$, it can be seen 
that this will
never happen for $R_{ip} <a_{i}$. 

Thus we can assume that $R_{ip}\geq a_{i}$. 
Given this if a rate-pair does not fully utilize the
interference free levels at transmitter $i$, it must be that
$R_{ic} < b_i$. Suppose that this is true for receiver $1$ and 
consider increasing
$R_{1c}$ and decreasing $R_{1p}$ by the same amount until
$R_{1c} = b_{1}$. Note that since $a_{1} +
b_{1} = L_1$ and $R_{ic} + R_{ip} \geq L_1$, 
when we decrease $R_{1p}$ in this way it will never
cause it to become less than $a_{1}$. 
Changing $R_{1p}$ and $R_{1c}$ in this manner will not violate 
any of the constraints in (\ref{eq:MACi}) at receiver $1$, 
since every constraint involving $R_{1c}$ also involves $R_{1p}$. 
If this can be done without
violating the first or second
constraints at receiver $2$ then we are done. Otherwise, it must
 be that at least one of these constraints becomes 
 tight when $R_{1c}$ reaches the 
value $R_{1c}^* =b_{1} -\Delta$, for
 some $\Delta >0$. Note that
\begin{equation}\label{eq:const1}
\begin{split}
R_{2p} + R_{2c} + R_{1c}^* &= R_2 + R_{1c}^*\\
& \leq U_2 + b_1 -\Delta \\ 
&= \max(n_{22},n_{21}) -\Delta
\end{split}
\end{equation}
and so the first constraint at receiver 2 can not be tight.
This implies that the
second constraint at receiver 2 must be tight, i.e.,
\begin{equation}\label{eq:r2p}
R_{2p} + R_{1c}^* =
 \max(n_{21},n_{22}-n_{12}).
\end{equation}
Combining this with (\ref{eq:const1}) we have
\begin{align*}
R_{2c} &\leq \max(n_{22},n_{21}) - \max(n_{21},n_{22}-n_{12}) - \Delta\\
& = b_2 - \Delta.\\
\end{align*}
It then follows that $R_{2p} \geq a_{2} + \Delta$. 
In other words user 2's interference free levels must be
 ``underutilized'' by at least as much as user 1's.

Now consider increasing $R_{1c}$ from $R_{1c}^*$ by $\Delta$ while
simultaneously reducing $R_{1p}$ and $R_{2p}$ each by $\Delta$ and
also increasing $R_{2c}$ by $\Delta$. The above calculations show that
no constraint will be violated if we only changed $R_{1c}$, $R_{1p}$
and $R_{2p}$ in this way. Likewise, by applying the same
argument to $R_{2c}$, $R_{2p}$ and $R_{1p}$, changing these values 
will not violate any constraints. The only possible violation
could occur in the first constraints at either MAC, which involves
both $R_{1c}$ and $R_{2c}$. However this constraint also involves one
of the users' private rates, and so cannot be violated by the same
argument as in (\ref{eq:const1}). After this transformation, 
the resulting rate-split will fully utilize the interference 
free levels at both receivers.
\end{IEEEproof}

Combining Lemmas \ref{lem:int}, \ref{lem:full} and \ref{lem:fullutil},
we have proven Theorem~\ref{prop:cne}. 

We also note that by direct
calculation it can be shown that $\mathcal
C_{NE}$ always contains at least one efficient point, i.e., one point that is sum-rate optimal.   Indeed it can be shown that for a symmetric channel,
for $\alpha \leq 2/3$, the only efficient point in $\CNE$
is the symmetric sum-rate optimal point, while for $\alpha \geq 2/3$
all sum-rate optimal points are in $\CNE$.

\section{The Gaussian IC}

In the previous section, we completely characterized the Nash
equilibrium region for the two-user linear-deterministic IC.
In this section, we show that an analogous result holds for
the Gaussian channel model within a one bit approximation.

\subsection{Gaussian Channel Model}
Here, our goal is to characterize
rates in $\CNE$ for two-user Gaussian interference channels
represented by (see Fig.~\ref{fig:model})
\begin{equation}
\begin{split}\label{eq:gauss}
y_1&= h_{11}x_1 + h_{12}x_2 + z_1\\
 y_2&= h_{21}x_1 + h_{22}x_2 + z_2
\end{split}
\end{equation}
where for $i = 1,2$, $z_i \sim \mathcal{CN}(0,1)$ and the input
$x_i\in \mathbb C$ is subject to the power
constraint $\text{E}[|x_i|^2] \leq P$.
 Following~\cite{ETW07}, for $i=1,2$, we parameterize this channel by the
signal-to-noise ratios $\snr_i =
P|h_{ii}|^2$ and the interference-to-noise ratios $\inr_{ij} =
P|h_{ij}|^2$. 

The characterization
of $\CNE$ in the linear deterministic case relied on knowing
the exact capacity region $\C$ for the deterministic IC and 
that any point in this region can be exactly achieved by a
non-randomized Han-Kobayashi schemes. For the Gaussian IC, 
$\C$ is only known in the case of
very weak~\cite{SKC07,AV08,MoK08} or very strong
interference~\cite{Car75,HK81}. Otherwise, $\C$ is not known exactly 
but in \cite{ETW07} it is characterized to ``within one bit'' for all
parameter ranges. Furthermore, \cite{ETW07} shows that in a 
general Gaussian IC, we can achieve any point within one bit 
by a non-randomized Han-Kobayashi scheme. These one-bit gaps will 
effect how accurately
we can characterize $\CNE$ in the Gaussian case. 
Namely, in general we will also be able to characterize this region
only to within a one bit gap (though for particular channels this gap
may be smaller).

\subsection{Main Results}

For the Gaussian IC our main result is to show an analogue to 
Theorem~\ref{prop:cne} that characterizes 
$\CNE$ to within one bit. In this case we will give both an inner
bound and outer bound on $\CNE$. Both of these bounds will be given in
terms of a capacity region and a ``box'' as in the deterministic case. 
The true capacity region will be used for the outer bound, while an 
achievable ``Han-Kobayashi'' region, $\C_{HK}$, will be used for the
inner bound. Here, $\C_{HK}$ corresponds to the set of rates 
that are achievable using the specific class of Han-Kobayashi 
schemes in \cite{ETW07} (this will be defined more precisely in 
Sect.~\ref{sec:achievG}). This region is within 1-bit of the 
capacity region $\C$, i.e.~if $(R_1,R_2) \in \C$, then 
$((R_1-1)^+,(R_2-1)^+ \in \C_{HK}$. 

The box $\B$ used for the outer bound 
is given by 
\[
{\mathcal B} = \{(R_1,R_2): L_i \leq R_i \leq U_i, 
\forall i = 1,2\},
\]
where for each user $i=1,2$
\[
L_i := \log\left(1+\frac{\snr_i}{1+\inr_i}\right),
\] 
and 
\begin{equation}\label{eq:UiG}
\begin{split}
U_{i} = \min & \left \{ \log (1 + \snr_i + \inr_{ij}) - \log
\left(1+\frac{\left[\snr_j - \max(\inr_{ji},\snr_j/\inr_{ij})\right
]^+}{1 + \inr_{ji} +
\max(\inr_{ji},\snr_j/\inr_{ij})} \right), \right . \\
& \left . \log(1+\snr_i)\right \}.\\
\end{split}
\end{equation}
While the inner bound is given in terms of the ``box,''
\[
\B^{-} = \{(R_1,R_2): L_i \leq R_i \leq \max(U_i-1,L_i), 
\forall i = 1,2\},
 \]
which differs from $\B$ by at most one bit.

We next state the analogous result to Theorem~\ref{prop:cne}.

\begin{theorem}\label{prop:cneG}
$\C_{HK}\cap B^{-} \subseteq \CNE \subseteq  \C \cap \B$. 
Moreover, for a Gaussian IC with strong interference, $\CNE = 
\C_{HK}\cap B$.
\end{theorem}

In certain cases, when we know additional
properties of the capacity region we can strengthen these results. 
For example for very weak interference, from the results in
\cite{SKC07,AV08,MoK08} it is known that the maximum sum-rate in $\C$ is 
achieved by simply treating interference as noise, which is also in
$\C_{HK}$.  This corresponds
exactly to the lower-left corner of $\B$ and $\B^{-}$. Hence,
Theorem~\ref{prop:cneG} implies that $\CNE$ contains the
single point $(L_1,L_2)$ and thus in this case $\CNE$ 
is characterized exactly.

We also note that the bounds in Theorem~\ref{prop:cneG} can be 
shown to be 
within a constant gap of the bounds given in Theorem~\ref{prop:cne}
for a related deterministic IC, which 
is obtained by the mapping $n_{ii} = \lfloor\log(\snr_i)\rfloor$ and 
$n_{ij} = \lfloor \log(\inr_{ij})\rfloor$.

In the next section we will give a 
proof of Theorem~\ref{prop:cneG} that is based on 
generalizing each of the steps we used in the deterministic case.

\subsection{Proofs}

\subsubsection{Non-equilibrium points}
We begin by showing that certain
rate-pairs can not be in $\CNE$.

\begin{lemma}\label{lem:L1}
If $(R_1,R_2) \in \CNE$, then
$R_i \geq L_i := \log(1+\frac{\snr_i}{1+\inr_i})$ for $i=1,2$.
\end{lemma}
\begin{proof}
Regardless of user $j$'s strategy, user $i$ can always achieve at
least rate $\log(1+\frac{\snr_i}{1+\inr_i})$ 
(with arbitrarily small probability of error)
by treating user $j$'s signal as noise. Hence, this is
always a possible deviation for user $i$ in any
$\epsilon$-game. Thus user $i$'s rate in any $\eta$-NE must
be at least $L_i - \eta$.
\end{proof}

The bound in Lemma~\ref{lem:L1} is a direct analog to the bound in
Lemma~\ref{lem:Li} for the linear deterministic channel, which
characterizes the lower bounds of the box $\mathcal B$. The next
lemma gives an upper bound corresponding to the bound in
Lemma~\ref{lem:Ui}.

\begin{lemma}\label{lem:UG}
If $(R_1,R_2) \in \CNE$, then
$R_i \leq U_i$, where $U_i$ is given in (\ref{eq:UiG}).
\end{lemma}

\begin{IEEEproof}
Suppose $(R_1,R_2) \in \CNE$.  Without loss of generality, let
us focus on user $1$ to show the upper bound on $R_1$.

We define first a parameter:
\begin{equation}
 \sigma_v^2 :=
\max\left ( \frac{\inr_{21}}{\snr_2}, \frac{1}{\inr_{12}} \right).
\label{eq:sigma}
\end{equation}

Consider first the case that $ \sigma_v^2 >1$: this corresponds to
the case in the deterministic channel when the interference from the
signal user 2 transmits at its interference-free levels appears
below noise level at receiver 1. In this case, there will be no
minimum amount of interference that user 2 will cause to user 1 at a
NE, and we simply bound $R_1$ by its point-to-point capacity:
\begin{equation}
\label{eq:u1}
 R_1 \le \log (1 +\snr_1).
\end{equation}

The case when $\sigma_v \le 1$ is more interesting and we need a
tighter bound on $R_1$. Fix $\eta
 > 0$ and arbitrary. Given a sufficiently small $\eps
 > 0$, there exist a $(1-\eps)$-reliable strategy pair $(s^*_1,
s^*_2)$ achieving the rate pair $(R_1,R_2)$ that is also a
$\eta$-NE. As remarked in Section \ref{sec:formulation}, we can
assume that in this nominal strategy pair, both users use a common
block length $N$.  Applying Fano's inequality to user $1$ for
average bit error probability, we get the bound, for any block $k$: 
\beq
\label{eq:fanoG} R_1 \le \frac{I(m_1; \vy_1|\omega_1)}{N} + \delta 
\eeq
where $\vy_1$ is user $1$'s received signal over the block, $\delta$
depends on $\eps$ and goes to zero as $\eps$ goes to zero, and 
$\omega_1$ denotes any common randomness in $\vx_1$.
 Note that
we drop the block indices of the message and the signals to simplify
notation. Now,
\begin{eqnarray*} & & \frac{1}{N} I(m_1; \vy_1|\omega_1)\\
& \le & \frac{1}{N} I(\vx_1;\vy_1|\omega_1)\\
&\leq & \frac{1}{N} I(\vx_1;\vy_1)\\
& =  & \frac{1}{N} \left [h(\vy_1) - h(\vy_1|\vx_1) \right]\\
& = & \frac{1}{N} \left [h(\vy_1)-h(\vz_1) - h(\vy_1|\vx_1)
+h(\vz_1)\right]\\
& \le & \log (1 + \snr_1 + \inr_{12}) - \frac{I(\vx_2;\tvy_1)}{N}
\end{eqnarray*}
where
$$\tvy_1 := h_{12} \vx_2 + \vz_1.$$
Combining this with the above inequality, we get: \beq
\label{eq:upperG} R_1 \le \log (1 + \snr_1 + \inr_{12}) -
\frac{I(\vx_2;\tvy_1)}{N}. \eeq

The term $I(\vx_2;\tvy_1)$ plays the role of $H(\vs_2)$ in the
linear deterministic case. We now seek a lower bound on
$I(\vx_2;\tvy_1)$. Applying Fano's inequality to user $2$, we get
\begin{eqnarray*} R_2   &\le & \frac{1}{N} I(m_2; \vy_2|\omega_2) + \delta\\
& \le & \frac{1}{N} I(\vx_2;\vy_2) +\delta\\
& = & \frac{1}{N} I(\vu_2,\vv_2; \vy_2) +\delta
\end{eqnarray*}
where
$$\vu_2 = \vx_2 + \vv_2$$
and $\vv_2 \sim \CN(0, \sigma_v^2 I_N)$ independent of everything
else, with $\sigma_v^2$ defined as in (\ref{eq:sigma}).

Now,
\begin{eqnarray*}
I(\vu_2,\vv_2; \vy_2) & = & I(\vu_2;\vy_2) + I(\vv_2;\vy_2|\vu_2)\\
& = & I(\vu_2;\vy_2) + I(\vv_2;\tvy_2|\vu_2)\\
& \le & I(\vu_2;\vy_2) + I(\vv_2;\tvy_2)
\end{eqnarray*}
where
$$ \tvy_2 = \vy_2 - h_{22} \vu_2 = h_{22} \vv_2 + h_{21} \vx_1 + \vz_2$$
and the last inequality above follows from the Markov chain $\vu_2 -
\vv_2 - \tvy_2$.

Combining this with the previous equation, we get: \beq
\label{eq:upper2G} R_2   \le \frac{1}{N} I(\vu_2;\vy_2) +
\frac{1}{N}I(\vv_2;\tvy_2) + \delta. \eeq

Let us now consider an alternative strategy $s'_2$ for user 2: a
superposition of two i.i.d.\ Gaussian codebooks, one with each
component of each codeword having variance $1-\sigma_v^2$, and one with each
component of each codeword having variance $\sigma_v^2$. The codes have
block
length $NK$. If we choose $K \rightarrow \infty$, then standard
random coding argument and the chain rule of mutual information
implies that this scheme can achieve a rate of:
$$ \frac{1}{N}I(\tvu_2,\vv_2;\vy'_2) = \frac{1}{N}I(\tvu_2;\vy'_2)
+\frac{1}{N} I(\vv_2;\tvy_2),$$
where $\tvu_2 \sim \CN(0,1-\sigma_v^2 I_N)$ and $\tvu_2$
and $\vv_2$ are independent, and
$$ \vy'_2 = h_{22}(\tvu_2+\vv_2) + h_{21} \vx_1 + \vz_2.$$
We have:
$$ \frac{1}{N} I(\tvu_2;\vy'_2) \ge
\log\left (1+\frac{\snr_2(1-\sigma_v^2)}{1+ \snr_2\sigma_v^2 +
\inr_{21}} \right),$$
using a worst-case Gaussian noise argument \cite{DC01} 
on $\vv_2$ and $\vx_1$.
This implies from (\ref{eq:upper2G}) that \beq
\label{eq:upper3G} \frac{1}{N} I(\vu_2;\vy_2) \ge
\log \left(1+\frac{\snr_2(1-\sigma_v^2)}{1+ \snr_2\sigma_v^2 +
\inr_{21}} \right)
- \eta \eeq by definition that we are operating at a $\eta$-NE.

Next we relate $I(\vx_2;\tvy_1)$ to $I(\vu_2;\vy_2)$ and complete the
argument.
\begin{eqnarray*}
& & I(\vx_2;\tvy_1) \\
& = & I(\vx_2 ;h_{12}\vx_2+ \vz_1)\\
& = & I(\vx_2; \vx_2+ \tvz_1), \qquad \tvz_1 \sim \CN(0,
\frac{1}{\sqrt{\inr_{12}}} I_N) \\
& \ge & I(\vx_2; \vx_2+ \vv_2), \qquad \mbox{since $\sigma_v^2 =
\max(\frac{\inr_{21}}{\snr_2}, \frac{1}{\inr_{12}}) \ge
\frac{1}{\inr_{12}}$}\\
 & = & I(\vu_2,\vx_2) \\
& \ge & I(\vu_2;\vy_2)\\
& \ge & N \left [\log \left(1+\frac{\snr_2(1-\sigma_v^2)}{1+
\snr_2\sigma_v^2 + \inr_{21}} \right) - \eta \right ] \qquad
\mbox{from (\ref{eq:upper3G}).}
\end{eqnarray*}

Substituting this into (\ref{eq:upperG}), we get the final result:
\begin{eqnarray*}
R_1 & \le & \log (1 + \snr_1 + \inr_{12}) - \log
\left(1+\frac{\snr_2(1-\sigma_v^2)}{1+ \snr_2\sigma_v^2 + \inr_{21}}
\right) +\eta\\
& = & \log (1 + \snr_1 + \inr_{12}) - \log \left(1+\frac{\snr_2 -
\max(\inr_{21},\snr_2/\inr_{12})}{1+  \inr_{21} +
\max(\inr_{21},\snr_2/\inr_{12})} \right) +\eta.
\end{eqnarray*}
Combining this with inequality (\ref{eq:u1}) and letting $\eta
\rightarrow 0$ yields the desired
result.
\end{IEEEproof}
\vspace{10pt}

\subsubsection{Achievable Nash Equilibrium}\label{sec:achievG}

The lemmas in the previous section provide an outer bound on
$\CNE$. In this section we give an inner bound on $\CNE$ by
showing that this set contains $\C_{HK}\cap \B^{-}$. Motivated by the
deterministic analysis, we again consider a modified Han-Kobayashi 
scheme in which each user may send a private message, a common
message, and also a common random message that is generated using the
common randomness they share with their own receiver. Additionally, in
the Gaussian case, we allow a user to also send a private random
message using their common randomness. In the deterministic case,
sending such a message would not serve any purpose since a user's private
signal does not appear at all at the other receiver. In the Gaussian
case, a user's private signal is present at the other user and the
private random message is used to ensure that the effect of this signal is
essentially the same as ``noise.'' All of these messages
are again encoded using a superposition code, which we define formally
next.

For a given Gaussian IC, let $P_{ip}$ and $P_{ic}$
denote a user's private and common power respectively, where $P_{ip} +
P_{ic} = P$. 
As in \cite{ETW07}, we assume that $P_{ip}$ is set as follows:
\[
|h_{ji}|^2P_{ip}=
\begin{cases}
\min(1,\inr_{ji}), & \text{if $\inr_{ji} < \snr_{j}$}\\
0, & \text{otherwise}.
\end{cases}
\]
Let $\inr^p_{ji} = |h_{ji}|^2P_{ip}$ denote the $\inr$ at receiver $j$ due to
this choice of $P_{ip}$ and let $\snr^p_{i} = |h_{ii}|^2P_{ip}$ denote 
the corresponding $\snr$ at receiver
$i$.  Note that when $\snr_{j} > \inr_{ji} > 1$, 
the received interference power at receiver $j$ due to user $i$'s
private power is at the noise level.

As in the deterministic case, we define a randomized Han-Kobayashi
scheme to be one in which each user $i$ separates their message set
into $\mathcal M_{ip} \times \mathcal M_{ic}$ with rates $R_{ip}$
and $R_{ic}$, respectively, and also generates a random common
message set $\Omega_{ir}$ with rate $R_{ir}$. Additionally, we allow each
transmitter to generate a random private message set, $\Omega_{is}$ with rate 
$R_{is}$.  These messages are then
encoded using a superposition code as follows. First, the
transmitter encodes the common message $m_{ic}\in \mathcal M_{ic}$
and common random message $\omega_{ir}\in \Omega_{ir}$ into a codeword
$\vx_{ic}(m_{ic},\omega_{ir})$ from a codebook that satisfies the average
power constraint of $P_{ic}$. Given this codeword, it encodes the
private message $m_{ip}$ and private random message $\omega_{is} \in
\Omega_{is}$ into a codeword $\vx_{ip}(m_{ip},\omega_{is},\vx_{ic})$
from a codebook that is indexed by the common codeword $\vx_{ic}$
and satisfies the average power constraint of $P_{ip}$.  It then
transmits the superposition $\vx_{i} = \vx_{ic} + \vx_{ip}$. 
As in the deterministic case, we call such a scheme 
$(1-\epsilon)$-reliable if a user is able to decode both his 
common and private messages with reliability of $(1-\epsilon)$.

We again introduce a modified MAC region $\R_i^m$ for each
receiver $i$ that we will use to characterize the rates achievable by
such a scheme. This is the set of rate tuples 
that satisfy:
\begin{equation}\label{eq:MACiG}
\begin{split}
R_{ic} +R_{ip} + R_{jc} + R_{jr} &\leq \log\left(1+\frac{\snr_i + \inr_{ij} -\inr^p_{ij}}{\inr^p_{ij}+1}\right)\\
R_{ip}+ R_{jc} + R_{jr}  &  \leq \log\left(1 +\frac{\snr_{i}^p + \inr_{ij} - \inr^p_{ij}}{\inr^p_{ij}+1}\right)\\
R_{ip} & \leq \log\left(1+\frac{\snr^p_{i}}{\inr^p_{ij} +1}\right)\\
R_{ic} +R_{ip} &\leq  \log\left(1+\frac{\snr_i}{\inr^p_{ij}+1}\right).
\end{split}
\end{equation}
As in the deterministic case, these constraints arise from
considering the three user MAC region at receiver 1 corresponding to
the rates $R_{ic}$, $R_{ip}$ and $R_{jc} + R_{jr}$. 
Using these regions we then have the following characterization
of rate-splits that can achieved with this class of schemes.

\begin{lemma}\label{lem:randomizedG}
$\mathcal R_{RHK} = \mathcal R_{1}^m \cap \mathcal R_{2}^m$ is an
achievable region for randomized Han-Kobayashi schemes.
\end{lemma}

We define the Han-Kobayashi region $\C_{HK}$ in
Theorem~\ref{prop:cneG} as the set of rates $(R_1,R_2)$ for which
there exists a rate split
$(R_{1c},R_{1r},R_{1p},R_{1s},R_{2c},R_{2r},R_{2p},R_{1s})\in \R_{RHK}$
 with $R_1 =R_{1c} +R_{1p}$ and $R_{2} =
R_{2c} +R_{2p}$. Using \cite{ETW07}, it follows that $\C_{HK}$ is
within one-bit of $\C$.\footnote{The rate region studied
in \cite{ETW07} corresponds to the rates $R_{1}$ and $R_2$ that can be 
achieved with rate-splits in $\R_{RHK}$ where the 
common random and private random rates of both users are zero. It can be 
seen that the resulting region is equivalent to $\C_{HK}$ as defined here.}
Next we give an analog to Lemma~\ref{lem:decode_common}.

\begin{lemma}\label{lem:decode_commonG}
Any rate-tuple $(R_{1c}, R_{1r}, R_{1p},R_{1s},R_{2c},R_{2r},R_{2p},R_{2s})$
in the interior of $\R_{RHK}$ with $R_{ic} + R_{ip} \geq L_i$ for $i = 1,2$ 
can be achieved by a randomized Han-Kobayashi scheme in 
which each user $i$ decodes user $j$'s common and common random
message (with arbitrarily small probability of error).
\end{lemma}

The proof here follows exactly the same steps as in the deterministic case. In 
particular for $R_{ic} + R_{ip} \geq L_i$, note that
\begin{equation}\label{eq:fifthG}
R_{jc} + R_{jr} < \log\left(1 + \frac{\inr_{ij} - \inr_{ij}^p}{1 + \inr_{ij}^p}\right)
\end{equation}
which implies that if user $i$ can decode his own common and private
messages, he will be receiving user $j$'s common messages at a rate less
than the capacity of the channel over which these messages are sent. 

Note also that the private random rate $R_{is}$ does not show up in 
any of the constraints for $\R_{RHK}$. This is because in these constraints 
each user $i$ is treating the
other user's private message as worst-case noise with power
$\inr^p_{ij}$ and can remove their own random private message. The
next lemma shows that the random private rate can always be chosen so
that there is essentially no loss in this assumption.

\begin{lemma}\label{lem:priv_random}
Given any $\delta >0$ and any rate tuple in $\R_{RHK}$ with 
\[
R_{js} = (\log(1 + \inr^p_{ij}) - R_{jp} - \delta/2)^+
\]
then for a large enough block-length $N$, this rate tuple can
be achieved by a randomized Han-Kobayashi scheme such that 
\[
\frac{1}{N}h(h_{ij}\vx^p_j + \vz_1) \geq \log(\pi e(1+\inr^p_{ij})) -\delta.
\]
\end{lemma}

\begin{IEEEproof}
Given a rate tuple that satisfies the conditions of this lemma and a
constant $\delta >0$, 
we next describe a specific encoding of user $j$'s
private messages to ensure the conditions for this lemma hold true. Let $w_j =
(m_{jp},\omega_{js})$ denote the total private message to be encoded
by user $j$ (chosen from a total private codebook with rate $R_{jp} +
R_{js}$).  To encode this message\footnote{To keep the overall
  superposition code structure for our class of Han-Kobayashi schemes,
  we need to construct such a private codebook for each codeword $\vx_{jc}$ in user
  $j$'s common codebook. Here we focus on one such codebook.} we
consider the following two cases: (i) $R_{jp} \leq \log(1 + \inr^p_{ij})
- \delta/2$
and (ii) $R_{jp} > \log(1+\inr^{p}_{ij})$. In each case we will
separate $w_j$ into two messages so that $w_{j} =
(w_j^1,w_j^2)$.

{\it Case 1: $R_{jp} \leq \log(1 + \inr^p_{ij})-\delta/2$}. In this case we set 
$w_j^1 = m_{jp}$, i.e.~this is the private message which is to be
decoded at receiver $j$. Since the rate-tuple is in $\R_{RHK}$, this 
message must be decodable over a
Gaussian channel with a capacity of $\log(1+\snr_j^p)$. We then set
$w_{j}^ = \omega_{js}$, i.e., this is the private random message sent
by user $j$. By assumption this message will have a rate of 
$R_{js} = \log(1 + \inr^p_{ij}) - R_{jp} - \delta/2$. 

By choosing $N$ large enough, there will
exist a Gaussian broadcast codebook for these messages so that for a
given reliability, $w_{j}^1$ and $w_j^2$ can be received
reliably over the Gaussian channel given by 
\[
y = h_{ij}x_j^p + z_i
\]
where $x_{j}^p$ has average power $P_{jp}$ and the noise variance is
1, and $w_{j}^1$ can be received at user $j$'s receiver given
$w_{j}^2$ (equivalently given the private random message
$\omega_{js}$). 
Note that the first channel has a capacity of
$\log(1+\inr_{ij}^p)$. 

By applying Fano's inequality to the first receiver, we have that for
a large enough reliability we can find a block-length $N$ so that 
\[
N(\log(1+\inr_{ij}^p) - \delta/2) \leq I(w_j;\vy) + N\delta/2, 
\]
where $\vy = h_{ij}\vx^p_j + \vz_1$ denotes the received signal 
over a block of length $N$.

Now,
\begin{align*}
I(w_j;\vy) & = h(\vy) - N \log(\pi e).
\end{align*}
Hence, we have
\[
h(\vy) \geq N (\log(\pi e (1+ \inr^p_{ij})) - \delta)
\]
as desired.

{\it Case 2: $R_{jp} > \log(1+\inr^{p}_{ij})$.} In this case we
set $R_{js} =0$ and  choose $w_j^1$ and $w_j^2$ so that $m_{jp} =
(w_j^1,w_j^2)$ where 
$w_j^1$ is chosen from a code book with rate 
$\log(1+ \inr^p_{ij}) - \delta/2$ and $w_{j}^2$ is chosen from a
codebook with rate $R_{jp} - \log(1+ \inr^p_{ij}) + \delta/2$.

By choosing $N$ large enough, there will
exist a Gaussian broadcast codebook for these messages so that for a
given reliability, $w_{j}^1$ and $w_j^2$ can be
received at user $j$'s receiver, and given $w_j^2$, $w_j^1$
can be received reliably over the Gaussian channel  
\[
y = h_{ij}x_j^p + z_i.
\]

Applying Fano's inequality at the second receiver we have
\[
N R_{jp} \leq I(w_j;\vy,w_j^2) + N\delta/2.
\]

Now,
\begin{align*}
I(w_j;\vy,w_j^2) & = I(w;w_j^2) + I(w;\vy|w_j^2)\\
&= N(R_{jp} - \log(1+\inr_{ij}^p) + \delta/2) \\
& \; + h(\vy|w_{j}^2) - N \log(\pi e). 
\end{align*}
Hence, we have
\[
h(\vy|w{j}^2) \geq N (\log(\pi e (1+ \inr^p_{ij})) - \delta).
\]
Dropping the conditioning and dividing by $N$, it follows that
\[
\frac{1}{N}h(\vy) \geq \log(\pi e(1+ \inr^p_{ij}) - \delta
\]
as desired.

\end{IEEEproof}

Note that under the given power constraints, $N\log(\pi e(1+
\inr^p_{ij}))$ is the maximum possible value for $h(\vy)$
which is achieved when $\vy$ is a sequence of i.i.d.~Gaussian random
variables. Hence, this lemma
can be viewed as showing that when the private rates are sufficiently
high, $h(\vy)$ is well approximated by simply viewing $\vy$ as i.i.d.\
Gaussian.

We say that a rate-split $\vR$ is
{\it self saturated} at receiver $i$ if $\vR \in \R_{i}^m$ 
and any other choice of
$R_{ic}$ and $R_{ip}$ in which  $R_{ip} +R_{ic}$
is increased  (keeping all other rates fixed) will
result in a rate-split that is not in $\R_{i}^m$. 
Similar to the deterministic case, it can be
 shown that if receiver $i$ is
self-saturated and $R_{ic}+ R_{ip} \geq L_i$ then 
$R_{i} = R_{ip} + R_{ic}$ and $R_{jc} + R_{jr}$
must be inside the following two user MAC region:
\begin{equation}\label{eq:MACi2G}
\begin{split}
R_{i} + R_{jc} + R_{jr}  &\leq \log\left(1+\frac{\snr_i + \inr_{ij}
  -\inr^p_{ij}}{\inr_{ij}^p +1}\right)\\
R_{i} & \leq \log\left(1+\frac{\snr_i}{\inr_{ij}^p +1}\right)\\
R_{jc} + R_{jr} &\leq \log\left(1 + \frac{\inr_{ij} - \inr_{ij}^p}{1 + \inr_{ij}^p}\right).
\end{split}
\end{equation}
Moreover, if $R_{i}$ is increased then
this rate pair will no longer be in this region. Using this we have the
next lemma which gives an analogous result to Lemma~\ref{lem:int}
for the deterministic channel.

\begin{lemma}\label{lem:intG}
If there exists a rate tuple $\vR$
that is self-saturated with $R_{ic} + R_{ip} \geq L_i$
for both both receivers $i$, then
$(R_{1p}+R_{1c},R_{2p}+R_{2c})\in \CNE$.
\end{lemma}

\begin{proof}
The proof follows a similar argument as that for Lemma~\ref{lem:int}.
Given a rate-tuple $\vR= (R_{1c},R_{1r},R_{1p},R_{1s},R_{2c},R_{2r},R_{2p},
R_{2s})$ that satisfies the conditions in the lemma, then it follows 
from Lemmas~\ref{lem:randomizedG}, \ref{lem:decode_commonG}, and
\ref{lem:priv_random} that for any $\eta >0$ and $\epsilon >0$, 
there exists a randomized Han-Kobayashi scheme achieving rates 
$(R_{1c}-\eta/6,R_{1r}-\eta/6,R_{1p}-\eta/6,\tilde{R}_{1s},R_{2c}-\eta/6,
R_{2r}-\eta/6,R_{2p}-\eta/6,\tilde{R}_{2s})$ for which each receiver
can decode his own common and private messages as well as the other
user's common and common random messages with probability of error
less than $\epsilon$. Moreover, by possibly changing the private random rates
$\tilde{R}_{1s}$ and $\tilde{R}_{2s}$ to satisfy
Lemma~\ref{lem:priv_random} such a scheme can be found for
which 
\begin{equation}\label{eq:privR}
\frac{1}{N}h(h_{ij}\vx^p_j + \vz_1) \geq \log(\pi e(1+\inr^p_{ij})) -\eta/6,
\end{equation}
for each user $i$. 

Next we argue that for $\epsilon$ small enough such a pair of 
strategies must be a $\eta$-NE of a $\epsilon$-game.
Assume that these strategies are not 
an $\eta$-NE, and without loss of generality suppose that user $1$ can
deviate and improve his performance by at least $\eta$. After
deviating, the rates for the MAC
region at user 1 in (\ref{eq:MACi2G}) are given by
$\tilde{\vR} = (\tilde{R}_{1}, R_{2c} + R_{2r}-\eta/3)$, where
$\tilde{R}_1 \geq R_1 + 2\eta/3$. After this deviation,
the rates $\tilde{\vR}$
must violate either the first or second constraint in (\ref{eq:MACi2G})
for $i=1$ by at least $\eta/3$.

Suppose that user $1$ deviates to a blocklength $N$ strategy, which 
without loss of generality we can assume is the same as the original 
strategy. Then from Fano's inequality for the average bit
error probability it must be that
\begin{equation}\label{eq:tr1G}
\begin{split}
\tilde{R}_1 \leq \frac{I(\vx_1;\vy_1|\vx_{2c})}{N} + \delta
\end{split}
\end{equation}
where $\delta$ goes to zero as the average bit error probability
$\epsilon$ does.  Note that
\begin{align*}
\frac{I(\vx_1;\vy_1|\vx_{2c})}{N} & = \frac{1}{N}(h(\vy_1|\vx_{2c}) - 
h(\vy_1|\vx_1,\vx_{2c}) \\
& \leq \log(\pi e(1+ \snr_1 + \inr_{12}^p) - \log(\pi e(1+
\inr_{12}^p) + \eta/6 \\
& = \log\left(1+ \frac{\snr_1}{1 + \inr_{12}^p}\right) +\eta/6,
\end{align*}
where the second line follows from (\ref{eq:privR}).
Choosing $\epsilon$ small enough and combining this
with (\ref{eq:tr1G}) implies that
\[
\tilde{R}_1 <  \log\left(1+ \frac{\snr_1}{1 + \inr_{12}^p}\right) + \eta/3.
\]
Hence, the second constraint in (\ref{eq:MACi2}) can not be violated.

Likewise, since in the nominal strategy user $1$
was able to decode user $j$'s common and common random signals,
it follows that
\begin{equation}\label{eq:tr2G}
R_{2c} + R_{2p} \leq \frac{I(\vx_{2c};\vy_1)}{N} + \delta'.
\end{equation}
Combining (\ref{eq:tr1G}) and (\ref{eq:tr2G}) and
choosing $\epsilon$ small enough so that $\delta + \delta' <
\eta/6$, we have
\begin{align*}
\tilde{R_1} + R_{2c} + R_{2p} &\leq
\frac{1}{N}I(\vx_{2c},\vx_1;\vy_1) + \eta/6 \\
& \leq \log(\pi e(1+\snr_1 + \inr_{12}) - \log(\pi e(1+
\inr_{12}^p) + \eta/3\\
&= \log\left(1+ \frac{\snr_1 + \inr_{12} + \inr_{12}^p}{1
    +\inr_{12}^p}\right)
+\eta/3
\end{align*}
which shows that the first constraint in (\ref{eq:MACi2}) can not
be violated.

Therefore, such a deviation can not exist and the nominal
strategy must be a $\eta$-NE for small enough $\epsilon$.
Taking the limit as $\eta\rightarrow 0$,
it follows that the desired rates must lie in $\CNE$.
\end{proof}

Next, we turn to proving an analogue of Lemma~\ref{lem:full} for the
Gaussian model. To do this we need to define a parallel notion to 
the interference-free levels in the deterministic channel. In a Gaussian
channel, this again corresponds to the rate $L_i$, which can be
achieved by treating interference as Gaussian noise. We still want
to constrain both the common and private rates of each transmitter
so that this rate is utilized with as much ``common rate'' as possible.
Specifically, let
\[
a_{i} = \log\left(1+\frac{\snr^p_i}{1+\inr_{ij}}\right)
\]
be the required private rate at user $i$ and
\[
b_{i} = \log\left(1+\frac{\snr_i -
\snr^p_i}{1+\snr^p_i+\inr_{ij}}\right)
\]
be the required common rate at user $i$. The rate $a_i$ is the rate
achieved by user $i$'s private signal when treating the aggregate
interference plus noise as Gaussian noise, while the rate $b_i$ is
the rate achieved by user $i$'s common signal when treating its own
private signal plus interference plus noise as Gaussian noise.

We say that $(R_{ic},R_{ip})$ {\it 
fully utilizes the
  interference free rate} for user $i$ if
\begin{align}
R_{ip} &\geq a_i \label{eq:aiG}\\
R_{ic} &\geq b_i.  \label{eq:biG}
\end{align}

With this definition we have the following analogue of
Lemma~\ref{lem:full}.

\begin{lemma}\label{lem:fullG}
If $(R_{1c},R_{1p},R_{2c},R_{2p})$ fully utilizes the interference free
  rate for each user $i$, then there exists
random common rates $R_{1r},R_{2r} \ge 0$ such that
$(R_{1c},R_{1r},R_{1p},R_{1s},R_{2c},R_{2r},R_{2p},R_{2s})$ is  
self-saturated at both receivers for any choice of $R_{1s},R_{2s}$.
\end{lemma}

\begin{IEEEproof}
This proof parallels exactly the proof of Lemma \ref{lem:full}. We
will show that one can always increase $R_{2r}$ such that the
overall sum rate constraint (the first constraint in
(\ref{eq:MACiG})) is tight, 
so that receiver $1$ is self-saturated. Suppose no such 
choice of $R_{2r}$ exists. Then it must
be that $R_{2r}$ cannot be further increased because the
second constraint in (\ref{eq:MACiG}) is tight. 
However if this constraint is 
tight, then since $R_{1c} \geq b_1$,
it can be seen that the first constraint must also be tight.
\end{IEEEproof}

To complete the generalization of the deterministic case, we need a
parallel result to Lemma~\ref{lem:fullutil}, which we state next.

\begin{lemma}\label{lem:fullutilG}
For any point $(R_1,R_2) \in \C_{HK} \cap \B^-$ there exists a
non-randomized Han-Kobayashi rate-split that fully 
utilizes the interference-free levels at each transmitter.
\end{lemma}

\begin{IEEEproof}
As in the deterministic case, we begin with an arbitrary non-randomized
Han-Kobayashi rate-split and show that this can always be transformed
into one that fully utilizes the interference free levels at each
receiver. Note that again since
$a_{i} + b_{i} = L_i$, for any point in $\B^{-}$ there will
always be a sufficient amount of ``rate'' available to meet the
constraints in (\ref{eq:aiG}) and (\ref{eq:biG}).

First, we show that if $R_{ip} < a_{i}$ for either user $i$,
then we can always increase $R_{ip}$ and decrease $R_{ic}$ by the 
same amount until
$R_{ip} = a_{i}$. The only way such a transformation could not be done is
if the second constraint in (\ref{eq:MACiG}) prevented it. But by combining
(\ref{eq:fifthG}) with this constraint, it can be seen 
that this will
never happen for $R_{ip} <a_{i}$. Moreover, since $R_{ip}$ does not
appear in any of the constraints in (\ref{eq:MACiG}) at receiver $j$, 
such a change will never
result in any of those constraints being violated.

Thus we can assume that $R_{ip}\geq a_{i}$. 
Given this, if a rate-pair does not fully utilize the
interference free levels at transmitter $i$, it must be that
$R_{ic} < b_i$. Suppose that this is true for receiver $1$ and 
consider increasing
$R_{1c}$ and decreasing $R_{1p}$ by the same amount until
$R_{1c} = b_{1}$. Note that since $a_{1} +
b_{1} = L_1$ and $R_{ic} + R_{ip} \geq L_1$, 
when we decrease $R_{1p}$ in this way it will never
cause it to become less than $a_{1}$. 
Changing $R_{1p}$ and $R_{1c}$ in this manner will not violate any of the
generalized MAC constraints at receiver $1$, since every constraint in (\ref{eq:MACiG})
involving $R_{1c}$ also involves $R_{1p}$. If this can be done without
violating any constraints at receiver $2$ then we are done. Otherwise, it must
 be that at least one of the constraints at receiver $2$ involving
 $R_{1c}$ becomes tight when $R_{1c}$ reaches the 
value $R_{1c}^* =b_{1}-\Delta$, for
 some $\Delta >0$. The possible constraints
 here are the first and second. By the definition of $\B^{-}$ we have that 
$R_{2p} + R_{2c} \leq U_2 -1$ and so
\begin{align}
R_{2p} + R_{2c} + R_{1c}^* &= R_2 + R_{1c}^*\notag\\
& \leq U_2 -1 + b_1 -\Delta \notag\\ 
& \leq \log(1+\snr_2 +\inr_{21}) -1 -\Delta \label{eq:L15G}\\
&\leq \log\left(1+\frac{\snr_2 + \inr_{21} -\inr^p_{12}}{\inr^p_{12}+1}\right)-\Delta.\notag
\end{align}
This shows that the first constraint can not be tight. Note that the last 
inequality followed since $\inr^p_{12} \leq 1$.
This implies that the
second constraint must be tight, i.e.,
\begin{equation}\label{eq:r2pG}
R_{2p} + R_{1c}^* = \log\left(1 +\frac{\snr_{2}^p + \inr_{21} - \inr^p_{21}}{\inr^p_{21}+1}\right).
\end{equation}
Proceeding as in (\ref{eq:L15G}) we have
\begin{align*}
R_{2p} + R_{2c} + R_{1c}^* & \leq \log(1+\snr_2 +\inr_{21}) -1 -\Delta
\\
&\leq \log\left(1+\frac{\snr_2 + \inr_{21} -\inr^p_{12}}{\inr^p_{21}+1}\right)-\Delta.
\end{align*}
Combining this with (\ref{eq:r2pG}) we have
\begin{align*}
R_{2c} &\leq \log\left(1+ \frac{\snr_2 -\snr_2^p - \inr_{12}^p}{1 +
    \snr_2^p + \inr_{21}}\right) - \Delta \\
& \leq b_2 - \Delta. 
\end{align*}
And so it must also be that $R_{2p} \geq a_2 + \Delta$.
Now consider increasing $R_{1c}$ from $R_{1c}^*$ by $\Delta$, while
simultaneously reducing $R_{1p}$ and $R_{2p}$ each by $\Delta$ and
also increasing $R_{2c}$ by $\Delta$. As in the deterministic case,
the above calculations show that we will not violate any of the
constraints in the modified MAC regions at either receiver when doing this.
After this transformation, the resulting rate-split will 
fully utilize the interference free levels at both receivers.
\end{IEEEproof}

Combining the previous lemmas we have shown that all points in
$\C_{HK}\cap \B^{-}$ are in $\CNE$, proving the first part of 
Theorem~\ref{prop:cneG}. Applying the next lemma will complete this 
proof by generalizing Lemma~\ref{lem:fullutilG} for strong IC 
and showing that in that case the conclusions apply for all points in 
$\C \cap \B$.

\begin{lemma}
For a strong IC,
any point $(R_1,R_2) \in \C \cap \B$ can be achieved by a 
non-randomized Han-Kobayashi rate split that fully 
utilizes the interference-free levels at each transmitter.
\end{lemma}
\begin{IEEEproof}
Recall, the for a strong interference channel, we set $\inr_{ij}^p =0$
for each user $i$ and so we have $a_{i} = 0$ and $b_{i} = L_i$. It
follows that for any point $(R_1,R_2) \in \C_{HK} \cap \B$ there will 
be exactly one non-randomized Han-Kobayashi rate split that achieves
this rate, namely the one with $R_{1c} = R_1$, $R_{2c} = R_2$ and
$R_{1p} = R_{2p} =0$. This will trivially fully utilize the
interference-free levels at each transmitter. Furthermore, for such a
channel $\C_{HK} = \C$, completing the proof.
\end{IEEEproof}

As in the deterministic case, it can also be shown by direct
calculation that $\CNE$ will always contain at least one point that is
sum-rate optimal to within 1 bit. 

\section{Conclusions}

We have formulated a new information theoretic 
notion of a Nash equilibrium region
for interference channels. Moreover, we have used this notion to 
characterize the equilibria in both deterministic and Gaussian 
ICs. In the deterministic case we are able to exactly specify the Nash
equilibrium region, while in the Gaussian case we characterize the 
Nash equilibrium region to within one bit. The
analysis for the Gaussian case directly parallels our analysis for the
deterministic case, and thus serves as another illustration of the
utility of deterministic models in providing useful insights for the
more complicated Gaussian setting. 

Our approach here is based on assuming that a given transmitter and
the intended receiver share a source of common randomness. However, in
the case of deterministic channels it is shown in \cite{BeT08} that
this is not needed. Specifically, it is possible to achieve all points
in $\C\cap \B$ by time-sharing among structured schemes which do no 
coding over time and use no common randomness. The key property of these 
structured schemes is that the common signals of
the two users are {\em segregated} into separate levels at each of the
receivers (in contrast to the random coding schemes considered in this
paper, where the common signals of the two users are all mixed up.)  Each
transmitter may still use randomness to send a jamming signal
on specified levels, but by aligning the jamming signal  with the
interfering common signal  at the node's own receiver, the receiver does not
need to decode it. Hence, the receiver needs not know the random bits
generating the jamming signal. Such schemes can likely also be
translated to the Gaussian settings by using structured codes instead
of the Gaussian Han-Kobayashi schemes considered here.

The games we were considering here were games with full information,
i.e., each user has perfect knowledge of all channel gains as well as
the code-books of the other user. One possible future direction for
this work would be to relax this assumption and consider games with
incomplete information. Another natural direction would be to consider
interference networks with more than 2 users. Some preliminary work
in this direction for deterministic channels is given in \cite{BS10}
where it is shown that with more than two user efficient equilibria
may no longer exist.

\section*{Acknowledgment}

This work has  benefited from convergence on a common game-theoretic
formulation for interference channels  that also appears in
\cite{YTL08}, which studies communication with secrecy constraints.

\end{document}